\documentclass[format=acmsmall, review=false]{acmart}
\usepackage{acm-ec-21}
\usepackage{booktabs} % For formal tables
\usepackage[ruled,vlined,linesnumbered]{algorithm2e} % For algorithms

\SetAlFnt{\small}
\SetAlCapFnt{\small}
\SetAlCapNameFnt{\small}
\SetAlCapHSkip{0pt}
\IncMargin{-\parindent}

\setcitestyle{acmnumeric}

\usepackage{pst-node}
\usepackage{pst-coil}
 \usepackage[english]{babel}
 \usepackage{hyperref}
 \hypersetup{
     colorlinks=true,
     linkcolor=blue,
     citecolor=blue,
 }
\usepackage{bm,xfrac,amsmath,amsthm,soul, comment, xspace, mathtools}
\usepackage[T1]{fontenc} 
\usepackage{footnote,dsfont}
\usepackage{nicefrac}
\usepackage{thmtools, float}
\usepackage{thm-restate}
\usepackage{graphicx}
\usepackage{pifont,multirow}
\usepackage{textcomp}
\usepackage{multicol}
\theoremstyle{plain}

\newtheorem{lemma}{Lemma}[section]
\newtheorem*{lemma*}{Lemma}

\newtheorem{theorem}{Theorem}[section]
\newtheorem{corollary}{Corollary}[section]

\newtheorem{claim}{Claim}[section]
\newtheorem*{claim*}{Claim}

\newtheorem{definition}{Definition}[section]
\DeclareMathOperator*{\argmax}{\arg\!\max}
\DeclareMathOperator*{\argmin}{\arg\!\min}
\usepackage{enumitem}

\newenvironment{proof-idea}{\paragraph{Proof Idea:}}{\hfill$\square$}

\newcounter{experiment}[section]

\newcommand{\A}{\mathcal{A}}
\newcommand{\N}{{\mathcal N}}

\newcommand{\Set}{{\mathcal S}}

\newcommand{\M}{{\mathcal M}}
\newcommand{\Vals}{(v_i)_{i\in\N}}

\newcommand{\MMSins}{(\N,\M,\Vals)}

\newcommand{\MMSidins}{(n,\M, v)}

\newcommand{\classNP}{{\sf NP}}
\newcommand{\classAPX}{{\sf APX}}

\newcommand{\bg}{{\sf BIG}}
\newcommand{\nbg}{{\sf SMALL}}
\newcommand{\bl}{{\sf LARGE}}
\newcommand{\bs}{{\sf MEDIUM}}

\newcommand{\Parti}{{\sf PARTITION}}

\newcommand{\EF}{\sf{EF1}}
\newcommand{\EFX}{\sf{EFX}}
\newcommand{\PO}{\sf{PO}}

\newcommand{\Prop}{\sf{Prop1}}
\newcommand{\MMS}{{\sf{MMS}}}
\newcommand{\NSW}{\sf{NSW}}

\newcommand{\aMMS}{{\sf{\alpha\text{-}MMS}}}

\newcommand{\aMMSg}{{\sf{GENERAL\text{ }\alpha\text{-}MMS}}}

\newcommand{\aeMMS}{{\sf{(\alpha-\epsilon)^+\text{-}MMS}}}

\newcommand{\gPO}{{\sf{\gamma\text{-}PO}}}
\newcommand{\optaMMS}{{\sf{OPT\text{-}\alpha\text{-}MMS}}}

\newcommand{\PTAS}{{\sf{PTAS}}}

\newcommand{\ASet}{{\mathbb A}}
\newcommand{\PA}{{{A^\pi}}}
\newcommand{\PB}{{{B^\pi}}}

\newcommand{\B}{\mathcal B}

\newcommand{\Bigi}{{\sf Big}}
\newcommand{\Sml}{{\sf Small}}
\newcommand{\epst}{\tfrac{\epsilon}{2}}
\newcommand{\eps}{{\epsilon}}

\newcommand{\tmu}{\tilde{\mu}}

\newif\ifsoda
\sodatrue

\let\oldnl\nl% Store \nl in \oldnl
\newcommand{\nonl}{\renewcommand{\nl}{\let\nl\oldnl}}% Remove line number for one line

\renewcommand{\nonl}{\renewcommand{\nl}{\let\nl\oldnl}}% Remove line number for one line
\long\def\symbolfootnote[#1]#2{\begingroup%
\def\thefootnote{\fnsymbol{footnote}}\footnote[#1]{#2}\endgroup}

\title{Indivisible Mixed Manna: On the Computability of MMS+PO Allocations}

\author{Rucha Kulkarni}
\affiliation{%
    \institution{University of Illinois at Urbana-Champaign}
    \city{Urbana}
    \state{IL}
    \country{USA}}
\email{ruchark2@illinois.edu}
\author{Ruta Mehta}
\affiliation{%
    \institution{University of Illinois at Urbana-Champaign}
    \city{Urbana}
    \state{IL}
    \country{USA}}
\email{rutameht@illinois.edu}
\author{Setareh Taki}
\affiliation{%
    \institution{University of Illinois at Urbana-Champaign}
    \city{Urbana}
    \state{IL}
    \country{USA}}
\email{staki2@illinois.edu}

\begin{abstract}
In this paper we initiate the study of finding {\em fair} and {\em efficient} allocations of an indivisible mixed manna: Divide $m$ indivisible items among $n$ agents under the popular fairness notion of {\em maximin share $(\MMS)$} and the efficiency notion of {\em Pareto optimality $(\PO)$}. A mixed manna allows an item to be a good for some agents and a chore for others, and hence strictly generalizes the well-studied goods (chores) only manna. For the goods manna, non-existence of an $\MMS$ allocation prompted a series of works on finding approximate $\MMS$ allocations, and the best factor known to date is $\alpha=${\tiny $\sim$ }$\!\nicefrac{3}{4}$, while non-existence is only known for $\alpha$ close to $1$. The problem of finding $\aMMS$ allocation for the {\em (near) best $\alpha\in (0,1]$} for which it exists, remains unresolved even when the number of agents is a constant, while the problem of finding $\aMMS+\PO$ allocation is unexplored for {\em any} $\alpha\in (0,1]$. %We first observe that unlike good (chore) manna, $\aMMS$ allocation may not exist for any $\alpha>0$.

We make significant progress on the above questions for the case of mixed manna. First, we show that for any $\alpha>0$, an $\aMMS$ allocation may not always exist, thus ruling out solving the problem for a fixed $\alpha$.
Second, towards computing $\aMMS+\PO$ allocation for the best possible $\alpha$, we obtain a dichotomous result: We derive two conditions and show that the problem is tractable under these two conditions, while dropping either renders the problem intractable. The two conditions are: $(i)$ number of agents is a constant, and $(ii)$ for every agent, %her total value for all the goods and her total absolute value for all the chores are not {\em too close}. 
her absolute value for all the items is at least a constant factor of her total (absolute) value for all the goods {\em or} all the chores. 
In particular, first, for instances satisfying $(i)$ and $(ii)$ we design a $\PTAS$ -- an efficient algorithm to find an $(\alpha-\eps)$-$\MMS$ and $\gPO$ allocation when given $\eps,\gamma$ {\small $>0$}, for the highest possible $\alpha\in (0,1]$. Second, we show that if either condition is not satisfied then finding an $\alpha$-$\MMS$ allocation for {\em any} $\alpha\in(0,1]$ is $\classNP$-hard, even when a solution exists for $\alpha=1$. On $m$ item instances our $\PTAS$ runs in time $2^{O(1/\min\{\eps^2,\gamma^2\})}poly(\mbox{m})$ for given $\eps$ and $\gamma$, and therefore gives polynomial run-time for $\eps,\gamma$ as small as $O(\nicefrac{1}{\sqrt{\log m}})$.  

As corollaries, our algorithm resolves the open questions %\ST{cite} 
of designing a $\PTAS$ for a goods manna and a chores manna with constantly many agents to find an $\aMMS$ allocation for the best possible $\alpha$; the best known was $\alpha=$ {\tiny $\sim$ }$\!\nicefrac{3}{4}$ for goods manna, and $\alpha=\nicefrac{9}{11}$ for chores manna. To the best of our knowledge, ours is the first algorithm that ensures both approximate $\MMS$ and $\PO$ guarantees. In terms of techniques, for the first time, we use an LP-rounding through envy-cycle elimination as a tool to solve an $\MMS$ problem and ensure $\PO$, which may be of independent interest.
\end{abstract}

\begin{document}

\maketitle

\newcommand{\subprob}{{\bf Sub-problem}}
\section{Introduction}\label{sec:intro}
Finding fair and economically efficient allocations of indivisible items is a fundamental problem that arises naturally in various multi-agent systems \cite{steinhaus1948problem,BramsT96,Vossen02,moulin2004fair,EtkinPT05,budish2011combinatorial,ghodsi2018fair}, for example, school seats assignment, spectrum allocation, air traffic management, allocating computing resources on a network, splitting assets and liabilities in partnership dissolution, and office tasks. Many of these involve both goods that are freely disposable and chores that {\em have} to be assigned. In this paper we study the problem of finding fair and efficient allocations of a {\em mixed manna}, i.e., a set $\M$ of discrete items that are goods/chores, among a set $\N$ of agents with additive valuations. We note that a mixed manna allows an item to be a good (positively valued) for some agents, and a chore (negatively valued) for others, and thereby strictly generalizes the extensively studied goods (chores) manna (See Appendix \ref{sec:relWork} for a detailed discussion on related works).

To measure {\em fairness} and {\em efficiency} we consider the popular and well studied notions of maximin-share ($\MMS$) ({\em e.g.,} see \cite{budish2011combinatorial,KurokawaPW18,amanatidis2017approximation,ghodsi2018fair,FarhadiGHLPSSY19,garg2019improved}) and Pareto optimality ($\PO$) respectively. Pareto optimality is a sought after notion in economics, and when achieved means that there is no other allocation that makes all the agents better off and at least one of them strictly better off. %ensures that there is no other allocation where all agents are better off and one of them is strictly better off. 
The fairness notion of  maximin-share is inspired from the classical cut-and-choose mechanism\footnote{In case of divisible items and two agents, one agent cuts so that she is okay with both the bundles and the other person chooses (mentioned in the Bible).}. The $\MMS$ value of agent $i$ is the value that she can guarantee herself if she is to partition (cut) $\M$ into $n=|\N|$ bundles, given that she is the last agent to choose her favorite bundle. Naturally, she will try to maximize the minimum valued bundle in the partition. Formally, if $\Pi_n(\M)$ represents all possible partitions $(A_1,\dots,A_n)$ of $\M$ into $n$ bundles, and $v_i$ is her valuation function, then
\begin{equation}\label{eq:mms-def}
\MMS_i(\M) = \displaystyle\max_{(A_1,\dots,A_n) \in \Pi_n(\M)} \displaystyle\min_{k \in [n]} v_i(A_k) \ .    
\end{equation}

An {\em $\MMS$ allocation} is one where every agent gets at least her $\MMS$ value. The problem of finding an $\MMS$ allocation has seen extensive work in the case of a goods (chores) only manna, while no results are known for the mixed manna. Even for the goods (chores) manna, no work has explored the $\PO$ guarantee in addition to $\MMS,$ to the best of our knowledge; finding fair+(approximate) $\PO$ allocations has been studied for other fairness notions like $\EF$ and $\Prop$ \cite{BarmanKV18,Zeng2020,AzizMS20}. 
%CIW19, AleksandrovW20} study the problem respectively for $\EF,$ $\Prop$ and notions related to $\EF$. For the goods manna, there is a series of works on this (for $\EF,$ see {\em e.g.,} \cite{BarmanKV18,Zeng2020,chaudhury2020dividing,SandomirskiyS19} and for $\Prop,$ {\em e.g.,} \cite{AzizMS20}), or group fairness notions~\cite{Conitzer19}. 
In this paper we initiate the study of finding an $\MMS+\PO$ allocation for a mixed manna. 

For the goods manna, the notable result of Kurokawa, Procaccia and Wang \cite{KurokawaPW18} showed that an $\MMS$ allocation may not always exist, but $\alpha$-$\MMS$ allocations, where every agent gets at least $\alpha$ times her $\MMS$ value, exist for $\alpha=2/3$. This prompted works on efficient computation of an $\alpha$-$\MMS$ allocation for progressively better $\alpha\in[0,1]$ \cite{amanatidis2017approximation,BarmanK17,garg2018approximating, ghodsi2018fair}; the best factor known so far is $\alpha=(3/4 + 1/(12n))$ by Garg and Taki \cite{garg2019improved} for $n\ge 5$ agents. With a chores manna, $\MMS$ values are negative, and an $\aMMS$ allocation gives each agent $i$ a bundle of value at least $\frac{1}{\alpha}\cdot\MMS_i$. For this case too, starting from the work~\cite{AzizRSW17} for $\alpha=1/2,$ a series of works improved it to $9/11$ \cite{BarmanK17,HuangL19}. 

With a mixed manna we show that, for any fixed $\alpha\in(0,1],$ an $\aMMS$ allocation may not always exist (see Appendix \ref{sec:nonexist}); in contrast, non-existence with a goods manna is known for $\alpha$ close to one~\cite{KurokawaPW18}. This rules out efficient computation for any fixed $\alpha,$ and naturally raises the following problem.
\medskip

\noindent 
{\em Design an efficient algorithm to find an $\aMMS+\PO$ allocation for the best possible $\alpha$, i.e., the maximum $\alpha\in(0,1]$ for which it exists.}
\medskip

This {\em exact} problem is intractable: In the case of identical agents, an $(\alpha=1)$-$\MMS$ allocation exists by definition. However, finding one is known to be $\classNP$-hard for a goods manna.\footnote{Checking if a given instance admits an $\MMS$ allocation is known to be in $\classNP^\classNP$, but not known to be in $\classNP$ \cite{BouveretL16}.} On the positive side, a polynomial-time approximation scheme ($\PTAS$) is known for this case due to \cite{woeginger1997polynomial}; given a {\em constant} $\eps\in(0,1],$ the algorithm finds a $(1-\eps)$-$\MMS$ allocation in polynomial time. No such result is known when the agents are not identical. 
Guaranteeing $\PO$ in addition adds to the complexity, since even checking if a given allocation is $\PO$ is co$\classNP$-hard even with two identical agents \cite{AzizBLLM16}. In light of these results, we ask, 

%\begin{align*}
\noindent{\bf Question.} {\em Can we design a $\PTAS$, namely an efficient algorithm to find an $(\alpha-\eps)$-$\MMS+\gamma$-$\PO$ allocation, given $\eps,\gamma>0$, for the best possible $\alpha$?}
%\end{align*}

\noindent
\textbf{Our Contribution.} In this paper we make significant progress towards this question for mixed manna by showing the following dichotomy result:
We derive two conditions and show that the problem is tractable under these conditions, while dropping either renders the problem intractable. 
The two conditions are: $(i)$ number of agents $n$ is a constant, and $(ii)$ for every agent $i$, her total (absolute) value for all the items ($|v_i(\M)|$) is significantly greater than the minimum of her total value of goods ($v^+_i$) and her total (absolute) value for chores ($v^-_i$), i.e., for a constant $\tau>0$, $|v_i(\M)|\ge \tau\cdot \min\{v_i^+,v_i^-\}$.

In particular, first, for instances satisfying $(i)$ and $(ii),$ we design a $\PTAS$ (as asked in the above question). Second, we show that if either condition is not satisfied, then finding an $\alpha$-$\MMS$ allocation for {\em any} $\alpha\in(0,1]$ is $\classNP$-hard, even with {\em identical agents} where a solution exists for $\alpha=1$. This hardness is striking because it shows inapproximability within {\em any} non-trivial factor when either $(i)$ or $(ii)$ is not satisfied. This also indicates that the two conditions are unavoidable. 

Our algorithm, in principle, gives a little more than a $\PTAS.$ It runs in time $2^{O(1/\min\{\eps^2,\gamma^2\})}poly(\mbox{m})$ for given $\eps,\ \gamma$, thus gives polynomial run-time for $\eps,\gamma$ as small as $O(\nicefrac{1}{\sqrt{\log m}})$, where $m=|\M|$. 

\emph{$\aMMS + \PO$ for goods (chores) manna.} As a corollary, we obtain a $\PTAS$ for finding $\aMMS + \PO$ allocations of a goods manna and a chores manna when the number of agents is a constant. This improves the previous results for these settings in two aspects: $(i)$ provides the best possible approximation factor; factors better than the general case known for good manna are $4/5$ for $n=4$ \cite{ghodsi2018fair}, $8/9$ for $n=3$ \cite{GourvesM19}, and $1$ for $n=2$ \cite{BouveretL16}, and $(ii)$ provides an additional (approximate) $\PO$ guarantee. %We note that, previous works on finding fair+$\PO$ allocations also obtains $\gPO$. %approximate $\PO$ notion is standard, as \cite{BarmanKV18} also obti

\noindent
\textbf{Challenges.} The key challenge in solving this question is handling items of high value to any agent. In the goods or chores mannas, these items can be greedily assigned, for example as singleton bundles. But in a mixed manna, \emph{high valued} goods (chores) may have to be bundled with specific sets of chores (goods) or low valued items to form lesser valued bundles. 
Secondly, the $\MMS$ values of the agents, and the $\alpha$ for which $\aMMS$ allocation exist, both are not known. In fact, computing the exact $\MMS$ values is $\classNP$-hard (even with a goods manna). 

\emph{$\PTAS$ to find $\MMS$ values.} As the first key step for our main algorithm, we design a $\PTAS$ that returns $(1-\epsilon)$ approximate $\MMS$ values of agents, which may be of independent interest. %Most known works use such a $\PTAS$ to find $\MMS$ values of agents in a goods manna or a chores manna, namely the $\PTAS$ of \cite{woeginger1997polynomial}  for the goods manna, and of \cite{Jansen2016} for the chores manna. %\RK{maybe delete}Our $\PTAS$ may serve as an important tool to tackle $\MMS$ problems with mixed manna, as has been the case with the goods manna (\cite{woeginger1997polynomial} used in \cite{}). \Ruta{not sure if the last sentence is needed here. or is appropriate because our ptas is for instances satisfying (i) and (ii)}

%Such a $\PTAS$, known for the goods (chores) manna \cite{}, has served as an important tool to several works \cite{}. Similarly, While no such $\PTAS$ was known for the mixed manna.

\emph{A new technique to prove $\PO$.} Since certifying a $\PO$ allocation is a co\classNP-hard problem \cite{AzizBLLM16}, known works maintain a $\PO$ allocation with market equilibrium as a certificate \cite{BarmanKV18,GargM19,GargM20}. %, or consider maximum social welfare allocations (which give every item to the agent with the highest value for it) to ensure $\PO$ \cite{}.
We develop a novel approach to ensure $\PO$ with $\aMMS$ through LP rounding. The LP itself is intuitive, however the rounding is involved. It makes use of {\em envy-graph} and properties of the $\MMS$ in a novel way. This approach may be of independent interest. 
\medskip
% \vspace{-0.5cm}

\noindent{\bf Organization.} Section \ref{sec:prelims} gives a formal definition of the problem and notations.
Section \ref{sec:mixed-pos} discusses the main result of $\PTAS$ for the $\aMMS+\PO$ problem with the best possible $\alpha$; %, and briefly describes the NP-hardness results and the $\PTAS$ for computing $\MMS$ values. 
the formal proofs missing from this section due to space limitations are in Appendix \ref{sec:missing-proofs}. The formal and complete discussion of the $\PTAS$ for computing $\MMS$ values for the case when $\MMS\ge 0$ is in Section \ref{sec:identical} and for the $\MMS<0$ case is in Appendix \ref{appendix:ptas-neg-id}. Appendix \ref{sec:nonexist} discusses the non-existence of $\aMMS$ allocation for any $\alpha \in (0,1]$. Finally, the discussion of $\classNP$-hardness results is in Appendix \ref{sec:hardness}.

%\noindent \textbf{Organization.} In Section \ref{sec:prelims} we formally define the problem and some of it's properties. In Section \ref{sec:mixed-pos} we discuss our first main result, the $\PTAS$ (for non-identical agents). Section \ref{sec:hardness} shows the second main result on the hardness of approximation. Sections \ref{sec:mmsvals} and \ref{sec:ptas-neg-id}, discuss the $\PTAS$ for the identical agents $\aMMS$ problems for agents respectively with non-negative and negative $\MMS$ values. All missing proofs, unless otherwise stated, are in Appendix \ref{sec:missing-proofs}.

% \vspace{-0.2cm}
\section{Problem Definition and Notations}\label{sec:prelims}
% \vspace{-0.2cm}

\noindent{\em Notations.} We use $[k]$ to denote the set $\{1,2\cdots, k\}$. For $c\in \mathbb R$, $c^+$ denotes $\max\{c,0\}$.

We consider the problem of allocating a set $\M$ of $m$ indivisible items among a set $\N$ of $n$ agents in a {\em fair} and {\em efficient} manner, with the fairness notion of {\em maximin share $(\MMS)$} and the efficiency notion of {\em Pareto-optimality $(\PO)$}. Each agent $i\in \N$ has an additive valuation function $v_i:2^{\M}\rightarrow \mathbb R$ over sets of items. For a set $S\subseteq \M$, her value is $v_i(S) = \sum_{j\in S} v_{ij}$. Agents are called {\em identical} if their $v_i$s are the same function; in this case, the valuation function is denoted by $v$.

The set of items valued non-negatively (negatively) by an agent $i$ are called her {\em Goods} ({\em Chores}), and denoted by $\M^+_i=\{j\mid v_{ij}\ge 0\}$ ($\M^-_i=\{j\mid v_{ij}< 0\}$). The sets of all the goods and all the chores of the instance are defined as respectively $\M^+:=\cup_i \M^+_i,$ and $\M^-:=\M\backslash \M^+.$ We refer to an item $j$ as a \textit{good} if $v_{ij} \ge 0$ for \emph{some} agent and as a \textit{chore} if $v_{ij} < 0$ for all agents.  %\ST{This is not correct since $(\cup_i \M^+_i) \cap (\cup_i \M^-_i) \neq \emptyset$ the right way to write is: "$\M^-:=\cup_i \M^-_i$ and note that $\M^+ \cup \M^- \neq \emptyset$ ".} 

\noindent{\bf $\MMS$ values and allocation.} Let $\PA=\{A_1,A_2,\cdots, A_n\}$ denote a partition of all the items among the $n$ agents, referred as an \textit{allocation}, i.e., $A_i\cap A_{i'}=\emptyset$ for all distinct $i, i'$ in $\N$, and $ \cup_i A_i = \M$. And let $\Pi_n(\M)$ be the set of all possible allocations of $\M$ among $n$ agents. The maximin share $(\MMS)$ value of an agent $i$ is defined as \[\MMS^n_i(\M) = \displaystyle\max_{(A_1,\dots,A_n) \in \Pi_n(\M)} \displaystyle\min_{k \in [n]} v_i(A_k).\]    

We refer to $\MMS_i^n(\M)$ by $\MMS_i$ when the qualifiers $n$ and $\M$ are clear, and by $\MMS$ when agents are identical. Note that $\MMS_i$ can be negative too. 
%An allocation which gives every agent $i$ a set of items worth at least $\MMS_i$ is called an $\MMS$ allocation. %Such allocations may not exist~\cite{ProcacciaW14,kurokawa2016can}, hence 
%We define approximate $\MMS$ allocations as follows.

\begin{definition}[$\aMMS$ allocation]
\label{def:mms}
$\PA$ is called an $\aMMS$ allocation for an $\alpha\in (0,1]$, if for each agent $i\in\N$ we have $v_i(A_i) \ge \alpha \MMS_i$ if $\MMS_i\ge 0$, $v_i(A_i) \ge (1/\alpha) \MMS_i$, if $\MMS_i <0$. 
Equivalently, $ v_i(A_i) \ge \min \{ \alpha \MMS_i, (1/\alpha) \MMS_i  \}.$ When $\alpha\le 0,$ for simplicity, we define any allocation as $\aMMS.$
\end{definition}
%\noindent{\bf Pareto-optimal ($\PO$) and Pareto dominating allocations.} An allocation $\PA$ is said to be $\PO$ if there does not exist any $\PB\in \Pi_n(\M),$ called an allocation \textit{Pareto dominating} $\PA$, such that $\forall i\in \N, \ v_i(B_i) \ge v_i(A_i)$ and at least one of these inequalities is strict. It is easy to see that if there exists an $\aMMS$ allocation for a given instance then there is one that is both $\aMMS$ and $\PO$.
\noindent
{\bf $\gamma$-Pareto optimal $(\gPO)$ and $\gamma$-Pareto dominating allocations.}\label{def:gpo}
An allocation $\PA$ is said to be $\gPO$ if there does not exist any $\PB\in \Pi_n(\M),$ called an allocation $\gamma$-Pareto dominating $\PA,$ such that $\forall i\in \N, \ v_i(B_i) \ge (1+\gamma)v_i(A_i)$ if $v_i(A_i)\ge 0,$ and $v_i(B_i) \ge \frac{1}{(1+\gamma)}v_i(A_i)$ if $v_i(A_i)< 0,$ and for at least one $i$ the inequality is strict. 

An allocation is called $\PO$ if it is $0$-$\PO.$ It is easy to see that if there exists an $\aMMS$ allocation for a given instance then there is one that is both $\aMMS$ and $\PO$ (and thereby also $\gPO$). This is because if an allocation $\PB$ Pareto dominates an $\aMMS$ allocation $\PA$, then $\PB$ is also $\aMMS$. 

Since the problem of finding $\aMMS$ allocation is $\classNP$-hard for any $\alpha\in (0,1]$, we design a $\PTAS$ to compute an $\aMMS+\PO$ allocation for a sub-class of instances. To characterize this sub-class, we will need the following definition.

% \vspace{-1.5em}
\noindent
\begin{equation}
\label{eq:vipm}
\begin{array}{lr}
\text{For each agent }i\in \N,\text{ define  }\ \ v_i^+ =\sum_{j\in\M^+_i} v_{ij}\ \ \text{ and }\ \ v_i^- =\sum_{j\in\M^-_i} |v_{ij}|. \quad \quad & \quad\quad\quad\quad
\end{array}
\end{equation}
% \vspace{-2em}

\begin{definition}[$\aMMS+\PO$ Problem]
\label{def:amms}Given an instance $\MMSins$ and $\alpha\in (0,1]$ where, 
% \vspace{-0.5cm} 
\begin{enumerate}
    \item the number of agents $n$ is constant, and 

% \vspace{-0.3cm} 
    \item for some constant $\tau>0$, for every agent $i\in\N$, $|v_i(\M)|\ge \tau\cdot \min\{v_i^+,v_i^-\},$
\end{enumerate}
either find an allocation $\PA\in \Pi_n(\M)$ that is both $\aMMS$ and $\PO$, also called an $\aMMS+\PO$ allocation, or correctly report that such an allocation does not exist for the given instance. 
\end{definition}

The above problem without the $\PO$ guarantee is called the $\aMMS$ problem. Unlike the goods manna or the chores manna, for the mixed manna an $\aMMS$ allocation may not exist for \textit{any} $\alpha>0,$ as shown in Appendix \ref{sec:nonexist}. Therefore, for a mixed manna, we can only hope to find an $\aMMS$ allocation for the maximum possible $\alpha$ value for the given instance, formally defined below.

\begin{definition}[$\optaMMS+\PO$ Problem.]\label{def:optmms} Given an instance $\MMSins$, find an allocation which is $\aMMS+\PO$ for an $\alpha\in(0,1]$ such that there is no $\alpha'$-$\MMS$ allocation for any $\alpha'>\alpha$. 
\end{definition}

Note that, given an algorithm for the $\aMMS+\PO$ problem it is easy to solve $\optaMMS+\PO$ by doing a binary search on the value of $\alpha$. We design a $\PTAS$ for the former in Section \ref{sec:mixed-pos} and thereby solve the latter efficiently up to a small error. By $\PTAS$ we mean, given constants $\epsilon,
\gamma>0$, in polynomial-time it either returns an $(\alpha-\epsilon)^+$-$\MMS+\gPO$ allocation, or correctly reports that no $\alpha$-$\MMS$ allocation exists.

The following observations will be useful in what follows (See Appendix \ref{sec:missing-proofs} for proofs). %Lemma \ref{lem:mms-sign} helps to find the sign of $\MMS$ value of the agents, Lemma \ref{lem:avg} shows a trivial upper bound on the $\MMS$ values, and Lemma \ref{lem:scale} proves that the $\aMMS+\PO$ problem is scale-free. 

\begin{restatable}{lemma}{Sign}\label{lem:mms-sign}
$v_i(\M) \ge 0$ iff $\MMS_i \ge 0$.
\end{restatable}

\begin{restatable}{lemma}{Average}\label{lem:avg}
$\MMS_i \le v_i(\M)/|\N|$ for all $i\in \N$.
\end{restatable}

\begin{restatable}{lemma}{ScaleInv}\label{lem:scale}[Scale Invariance]
$\aMMS+\PO$ allocations for the instances $\MMSins$ and $(\N,\M,(v'_i)_{i \in \N})$ are the same when for all  $i,j,\  v'_{ij}=c_i \cdot v_{ij}$ for some constants $c_i>0$.
\end{restatable}

%\begin{lemma}[Scale invariance]\label{lem:scale}
 %Then we will have\begin{enumerate}
   % \item[$(i)$] $\MMS_i^{\text{new}} = c_i \cdot \MMS_i$ and
   % \item[$(ii)$] If $\PA$ is an $\aMMS$ allocation for $\MMSins$, then for $(\N,\M,(v_i^{\text{new}})_{i \in \N})$ too, $\PA$ is an $\aMMS$ allocation.
%\end{enumerate}
%\end{lemma}

%\input{3.hardness}
%\input{4.identical}
\section{$\PTAS$ for $\aMMS+\PO$ with Non-identical Agents}\label{sec:mixed-pos}
In this section, we present our main result, namely a $\PTAS$ for the $\optaMMS+\PO$ problem. %, and discuss the other results briefly due to space constraints. 

For $\optaMMS+\PO$ the crucial step is to get a $\PTAS$ for the $\aMMS+\PO$ problem, discussed next.
Recall that the definition of the latter, namely Definition \ref{def:amms}, assumes two conditions on the input instance $\MMSins$: $(i)$ number of agents is a constant, and $(ii)$ for each $i\in\N$, $|v_i(\M)|\ge \tau\cdot \min\{v_i^+,v_i^-\},$ where $\tau>0$ is a constant. Let us first briefly discuss why both of these conditions are unavoidable.
\medskip

\noindent{\bf Hardness of approximation.} 
In Appendix \ref{sec:hardness}, we show the following theorem by proving that if either condition is dropped then the problem is intractable for {\em any} $\alpha\in(0,1],$ even when exact $\MMS$ allocation exists. 

\begin{restatable}{thm}{GenMMSHard}\label{thm:hard}
For any instance $\MMSidins$ with identical agents and $v(\M)>0$ such that exactly one of the following two holds: (a) either $n=2$ or (b) $|v(\M)| \ge \tau\cdot  \min\{v(\M^+),|v(\M^-)|\}$ for a constant $\tau$, finding an $\alpha$-$\MMS$ allocation of $\MMSidins$ for any $\alpha \in (0,1]$ is $\classNP$-hard.
\end{restatable}

To prove the above theorem, we design two reductions from a well-known $\classNP$-hard problem {\Parti} to the problem of finding an $\aMMS$ allocation of an instance $\MMSidins$ for any $\alpha\in (0,1],$. %Both the reductions construct instances with identical agents, where the first one has exactly two agents while the second one satisfies $|v(\M)| \ge \tau\cdot \min\{v(\M^+),|v(\M^-)|\}$. 
The tricky part in these reductions is to guarantee that an $\aMMS$ allocation for {\em any} $\alpha>0$ maps to a solution of {\Parti}. 
\medskip

\noindent{\bf Computing the $\MMS$ values.}
The first step in our $\PTAS$ is to compute the $\MMS$ values of the agents, which is equivalent to finding an $\MMS$ allocation with identical agents. The above hardness result rules out even approximating the $\MMS$ values within any non-trivial factor in polynomial time if either condition is not satisfied. For the instances satisfying both, in Section \ref{sec:identical} we design an efficient algorithm to compute the $\MMS$ values up to a small multiplicative error. We need to tackle the cases with $\MMS\ge 0$ and $\MMS<0$ separately; note that the sign of $\MMS$ can be easily determined using Lemma \ref{lem:mms-sign}. Formally, we show the following (see Section \ref{sec:identical} and Appendix \ref{appendix:ptas-neg-id}). %Finding the sign of $\MMS$, then applying the $\PTAS$ for that case, formally prove the following theorem.

%\vspace{-0.5em}
\begin{restatable}{thm}{ThmIdentical}\label{thm:identical}
Given an instance $\MMSidins$ and a constant $\epsilon>0$, if $(i)$ $n$ is a constant and $(ii)$ for each $i\in [n]$, $|v_i(\M)|\ge \tau\cdot \min\{v_i^+,v_i^-\},$ where $\tau>0$ is a constant. Then, there is a $\PTAS$ to compute a $(1-\epsilon)$-$\MMS$ allocation. 
%Given an instance $\MMSidins$ and a constant $\epsilon>0$ there is a $\PTAS$ to compute a $(1-\epsilon)$-$\MMS$ allocation. 
\end{restatable}

%\medskip
Our $\PTAS$ for the $\aMMS+\PO$ problem takes as input the instance $\MMSins$, a parameter $\alpha\in (0,1]$, and constants $\epsilon,\gamma>0$, and it either finds an allocation that is $\aeMMS+\gPO$ allocation, or correctly reports that an $\alpha$-$\MMS$ allocation does not exist; the latter may very well be the case for {\em any} $\alpha\in (0,1]$ as shown in Appendix \ref{sec:nonexist}.

\noindent
{\bf Pre-processing.}
First, note that the problem is non-trivial only if $\alpha>\epsilon,$ otherwise since $(\alpha-\epsilon)^+=0,$ thus an allocation that gives every item to the agent with the highest value for it is $\aeMMS+\PO,$ and returned. Therefore, now on we assume that $\alpha>\epsilon$.

Next we re-define $\epsilon$ as $\min\{\epsilon, \frac{\gamma\alpha}{(1+\gamma)}\}.$ This is done for technical reasons to ensure that the final allocation is also $\gamma$-$\PO.$ It does not harm the $\MMS$ guarantee, as an $\aeMMS$ allocation with a smaller $\epsilon$ is also an $\aeMMS$ allocation with respect to the given $\epsilon.$ Note that when $\alpha$ and $\gamma$ are constants, so is the new value of $\epsilon$. Finally, we assume there are no agents with $v(\M)=0$. Note that because of condition $2$ of the problem, when $v(\M)=0$ then the value of every item for this agent is $0.$ Also note that their $\MMS=0.$ Thus, we can allocate all the chores arbitrarily among agents with $v(\M)=0,$ and remove them. It is easy to see that the $\MMS$ value of the remaining agents can only improve, and all $\aMMS$ allocations are retained, by the removal of all the chores and a subset of agents. The problem then reduces to a goods manna case with no agents with $v(\M)=0,$ which is solved as a special case of the $\PTAS$ we will describe. 

\medskip

Due to the pre-processing step, now on we assume that $\MMSins,$ the given fair division instance, satisfies $v_i(\M)\ne 0$ for every agent $i\in \N$. We first scale the valuations so that $|v_i(\M)|=n$ since the problem is scale free by Lemma \ref{lem:scale}. Without loss of generality, we assume that the given constants $\alpha,\epsilon,\gamma>0$ are such that $\alpha>\epsilon,$ and $\epsilon\le \frac{\gamma\alpha}{(1+\gamma)}.$ The algorithm first applies the relevant $\PTAS$ from Section \ref{sec:identical} or Appendix \ref{appendix:ptas-neg-id} to compute the $\MMS$ value of every agent approximately up to a factor $(1-\epsilon/2)$. If $\overline{\MMS}_i$ is the value returned by the algorithm for agent $i,$ we know $\overline{\MMS}_i\ge \min\{(1-\epsilon/2)\MMS_i, (1/(1-\epsilon/2))\MMS_i\}.$ The algorithm then tries to find an $(\alpha-\epsilon/2)^+$-$\overline{\MMS}_i$ allocation, and fails only when an $\aMMS$ allocation does not exist. 
\medskip

\noindent{\bf High-level Approach.}
At a high level, the algorithm to find an $(\alpha-\epsilon/2)^+$-$\overline{\MMS}_i$ allocation is as follows. We will classify all items as $\bg$, based on if they are highly valued by any agent relative to her $\MMS$ value, or $\nbg$ otherwise. Although the $\MMS$ values of agents can be arbitrarily small, we show that the number of $\bg$ items is a function of $n,$ hence constant from condition $1$. Therefore, we can efficiently enumerate all partitions of the $\bg$ items. 

For each partition, we allocate the $\nbg$ items by solving an LP and rounding its solution. The LP ensures a fractional solution where every agent gets at least an $\alpha$-$\overline{\MMS}$ valued bundle. Next, through a careful rounding, we show that if there is an $\aMMS$ allocation where the $\bg$ items are allocated according to the current partition, then the allocation of all items obtained after rounding the LP solution is $(\alpha-\epsilon/2)^+$-$\overline{\MMS_i}.$ Among all the fractional $\aMMS$ allocations found by combining some $\bg$ item partition with the allocation of $\nbg$ as per the LP solution, we find the one, say $\A=[\A_1,\cdots,\A_n],$ with the highest value for the sum of valuations of all the agents, i.e., $\sum_iv_i(\A_i).$ That is, we find a fractional allocation, 
\begin{equation*}
    \A \in \argmax_{B \in \Pi_n(\bg)} \max_{A_i \supseteq B_i, A \text{ is $\aMMS$}} \sum_{i \in \N} v_i(A_i).
\end{equation*}
Finally, we show that the rounded solution, call it $\A^r$, is $\gamma$-$\PO$, by showing that for an allocation to $\gamma$-Pareto dominate $\A^r$, it must be an $\aMMS$ allocation and have higher welfare than $\A$. This proof is quite involved and uses several new ideas, including the way we round the LP solution, to show Pareto optimality of the integral allocation\footnote{Recall that testing $\PO$ is co$\classNP$-hard \cite{AzizBLLM16}, and market equilibrium (or highest sum of valuations for the trivial $\PO$ allocation) is the only technique in all known literature to certify that an allocation is $\PO$.}.
 The following bound on the $\MMS_i$ values will be useful in the analysis, and follows from Lemma \ref{lem:avg} and $|v_i(\M)|=n,\ \forall i\in \N$. 

\begin{lemma}
\label{lem:mu1}
For each agent $i\in \N$, $\MMS_i \le 1$ if $v(\M) \ge 0$, otherwise $\MMS_i \le -1$.
\end{lemma}

In the remaining section, we will discuss the details and formalize these ideas, with some proofs moved to Appendix \ref{sec:missing-proofs} in order to convey the main ideas within the page limit. Also, for brevity, at times we will refer to $\overline{\MMS_i}$ as $\tilde{\mu}_i$ and to $\MMS_i$ as $\mu_i.$

\subsection{$\bg$ and $\nbg$ Items}
Next we classify items into sets $\bg$ and $\nbg$, and show bounds on the size of the $\bg$ items set.

\begin{definition}[Big and Small items]\label{def:big-sml}
The sets of all $\bg$ goods $(\bg^+_i)$ and $\bg$ chores $(\bg^-_i)$ of agent $i$ are defined as, 
\begin{equation*}
    \begin{split}
  &\bg^+_i:=\{j\in \M^+ \mid (\tilde{\mu}_i\ge 0\text{ and }v_{ij}>\epsilon\tilde{\mu}_i/(2n))\text{ or } (\tilde{\mu}_i< 0\text{ and }v_{ij}>\epsilon/(2n))\},\text{ and }\\
  &\bg^-_i:=\{j\in \M^-\mid -v_{ij}>\epsilon/(2n)\}.
    \end{split}
\end{equation*}
The union of all sets $\bg_i^+$ is called $\bg^+=\cup_i \bg^+_i,$ and of all $\bg_i^-$ sets is called $\bg^-=\cup_i \bg^-_i.$ Finally, the set of all $\bg$ items is called $\bg:=\bg^+\cup \bg^-.$ 

\noindent Any item that is not in $\bg$ is called a $\nbg$ item. We define $\nbg$ goods and chores for agent $i$ as $\nbg^+_i=\M_i^+\backslash \bg^+_i,$ and $\nbg^-_i=\M_i^-\backslash \bg^-_i.$ Similarly, the sets of $\nbg$ goods, $\nbg$ chores, and $\nbg$ items are respectively $\nbg^+=\M^+\backslash\bg^+,\ \nbg^-=\M^-\backslash\bg^-,$ and $\nbg=\nbg^+\cup\nbg^-$. \end{definition}

In the remaining section, we will show the size of $\bg$ is constant. For this, we make two useful observations, then show the bound on $\bg$.   
\begin{restatable}{claim}{clmMuTmu}\label{clm:mu_tmu}
For the approximate $\MMS$ values $\tmu_i,$ we have, if $\mu_i> 0,$ then $\tmu_i\in [(1-\epsilon/2)\mu_i, \mu_i],$ if $\mu_i=0$ then $\tmu_i=0$ and if $\mu_i < 0,$ then $\tmu_i\in [\mu_i/(1-\epsilon/2),\mu_i].$ 
\end{restatable}
The claim follows from the guarantees of Theorems \ref{thm:main-id-pos} and \ref{thm:main-id-neg}.
Next, recall the definitions of $v_i^+$ and $v_i^-$ from equation \eqref{eq:vipm}. The next claim follows from condition $2$ of the problem.

\begin{restatable}{claim}{clmViAren}\label{clm:vi_are_n}
For all agents $i,$ $v_i^+\le O(n), v_i^-\le O(n).$
\end{restatable}
%Next we show a bound on the number of $\bg$ items. 

Next lemma shows a bound on $|\bg|$ (proof in Appendix \ref{sec:missing-proofs}). For this we show the bound of $O(n^2/\epsilon)$ on $|\bg_i^+|$ and $|\bg_i^-|$ for each agent $i$.
Note that if $\tmu_i$ is big enough then it is easy to prove that $|\bg_i^+|$ is a constant. The difficulty is when $\tmu_i$ is arbitrarily small, in which case $|\bg_i^+|$ can potentially be large -- a tricker case. %This is a trickier case and requires a bit involved argument. %However, we manage to show the bound with a bit involved argument. We show that if $|\bg^+_i|$ is too large, then the value of $\mu_i$ should have been significantly higher. 
The bound on $|\bg_i^-|$ follows from the definition of $\bg^-_i$ together with Claim \ref{clm:vi_are_n}.

\begin{restatable}{lemma}{lemNidBgGoodChore}
\label{cor:big-const}
The number of big items, i.e., $|\bg|\le O(n^3/\epsilon)$.
%\label{lem:nid-bg-good-chore}
% The number of $\bg$ goods of every agent $i\in \N$ is $|\bg_i^+|, |\bg_i^-|\le O(n^2/\epsilon)$. 
\end{restatable}
%To prove the above lemma we show that $|\bg_i^+|, |\bg_i^-|\le O(n^2/\epsilon)$ for each agent $i$. Note that if $\tmu_i$ is big enough then it is easy to prove that $|\bg_i^+|$ is a constant. The difficulty is when $\tmu_i$ is arbitrarily small, in which case $|\bg_i^+|$ can potentially be large. We show that if $|\bg^+_i|$ is too large, then the value of $\mu_i$ should have been significantly higher. The bound on $|\bg_i^-|$ follows from the definition of $\bg^-_i$ together with Claim \ref{clm:vi_are_n}.
%\vspace{-0.5em}
%\begin{restatable}{lemma}{lemNidBgChore}
%\label{lem:nid-bg-chore}
%The number of $\bg$ chores of every agent $i\in\N$, i.e.,  $|\bg_i^-|\le O(n^2/\epsilon)$.
%\end{restatable}
%\vspace{-0.5em}
%Taking a union of $\bg^+_i$s and $\bg^-_i$s using Lemma \ref{lem:nid-bg-good-chore} implies the following.
%\begin{corollary}
%\label{cor:big-const}
%The number of big items, i.e., $|\bg|\le O(n^3/\epsilon)$.
%\end{corollary}

\subsection{LP for Allocating $\nbg$ Items, and Rounding}\label{sec:LP}
Given a partition $\PB=(B_1,\dots,B_n)$ of $\bg$ items, next we write an LP to find a {\em fractional} allocation of $\nbg$ items such that together with $\PB$ this allocation gives at least $\alpha$-$\overline{\MMS}$ value to every agent. If there exists an $\aMMS$ allocation where the $\bg$ items are allocated as per $\PB$ then we show that the LP has to be feasible. 

For every agent $i,$ denote by $c_i$ the value from $\nbg$ that $i$ needs for her bundle's value to be at least $\alpha\cdot\tmu_i$ if $\tmu_i\ge 0$ or $(\nicefrac{1}{\alpha})\cdot\tmu_i$ otherwise, i.e., $c_i=\min\{(\nicefrac{1}{\alpha})\tmu_i,\alpha\tmu_i\}-v_i(B_i).$ 
\begin{align}
     \max \hspace{.47cm} & \sum\limits_{i\in \N} \left(\sum\limits_{j\in \nbg_i^+}v_{ij}x_{ij}-\sum\limits_{j\in \nbg_i^-}|v_{ij}|x_{ij}\right) \label{lp:small}\\
    \text{s.t. }\hspace{.47cm} &\hspace{-.47cm} \sum\limits_{j\in \nbg_i^+}v_{ij}x_{ij}-\sum\limits_{j\in \nbg_i^-}|v_{ij}|x_{ij} \ge c_i,\quad  &&\forall i\in \N \label{lp:ci}\\
    & \sum\limits_{i\in\N} x_{ij} \le 1,\quad \quad  &&\forall j\in \nbg^+ \label{lp:goods}\\
    & \sum\limits_{i\in\N} x_{ij} \ge 1,\quad \quad  &&\forall j\in \nbg^- \label{lp:chores}\\
    & x_{ij}\ge 0,\quad \quad \quad  &&\forall i\in \N,j\in \M. \label{lp:vars}
\end{align}

We now prove two properties (Lemmas \ref{lem:allocation-acyclic} and \ref{lem:shared}) that will help in obtaining an integral $(\alpha-\epsilon/2)$-$\overline{\MMS}$ allocation of items from a fractional $\alpha$-$\overline{\MMS}$ allocation. Let us assume the LP has a solution, say $x=[x_{ij}]_{i\in\N,j\in\nbg}$. We define a bipartite graph, called the \textit{allocation graph}, corresponding to $x$ as follows. There is a vertex corresponding to each agent in $\N$ in one part of vertices, and to each item in $\nbg$ in the other part, and for all $i\in \N$ and $j\in\nbg,$ edge $(i,j)$ exists if $x_{ij}>0.$ We show the following property of the allocation graph. 

\begin{restatable}{lemma}{lemAcyclic}\label{lem:allocation-acyclic}
The allocation graph of \emph{any} LP solution $x$ can be made acyclic in such a way that in the allocation corresponding to the new graph, say $x'=[x'_{ij}]_{i\in\N,j\in\nbg},$ every agent receives a bundle of the same or better value as in $x$.
\end{restatable}
To prove the lemma, we show re-allocations can be done along any cycle in a certain way without any agent losing any value that eliminates at least one edge. For every cycle, we define a particular scaled valuation function, and define weights for the edges to reflect the values to agents from the adjacent items. Then we add and subtract weights in a certain way along the cycle, taking into consideration if the adjacent item is a good or a chore, so that the allocation corresponding to the new weights, or equivalently (scaled) values to agents, does not contain this cycle. 

The next lemma follows since an undirected, acyclic graph forms a tree. 

\begin{restatable}{lemma}{lemShared}\label{lem:shared}
%\begin{lemma}\label{lem:shared}
The number of shared items in any acyclic allocation graph is at most $n-1.$
%\end{lemma}
\end{restatable}

Next we describe the notion of envy graph \cite{lipton2004approximately}, a directed graph corresponding to any allocation, that will be used to round the LP solutions. 

\noindent
\textbf{Envy Graph and Cycle Elimination.}
Given a set of agents $\N$ and an \textit{integral} allocation $\A$ of a set of items among them, each node in the graph corresponds to an agent in $\N$. There is a directed edge $(i\rightarrow k)$ corresponding to agents $i$ and $k$ if agent $i$ values agent $k$'s allocation more than her own. It is shown in \cite{lipton2004approximately} that the allocation can be modified so that its corresponding envy graph is acyclic, and no agent's valuation decreases. This is done by giving each agent in a cycle the bundle of her successor. The graph is updated and the process repeated until all cycles are eliminated. This process can be done efficiently \cite{lipton2004approximately}. 

\begin{restatable}{claim}{clmSinkVal}\label{clm:sinkval} In an allocation of $\M$ among $n$ agents, every sink agent $i$ corresponding to an acyclic envy graph has value at least $1$ for her own bundle if $v_i(\M)>0,$ and at least $-1$ otherwise.
\end{restatable}

\noindent
\textbf{Rounding the LP.} Using Lemmas \ref{lem:allocation-acyclic} and \ref{lem:shared}, we first modify the allocation graph of the LP solution so that it is a forest graph with at most $n-1$ shared items. Let $S$ be the set of all the shared items, $S^{-\eps}$ the set of all the shared chores whose absolute value is more than $\epsilon|\tmu_i|/(2n)$ for at least one agent, that is, $S^{-\eps}:= \{j \in S \mid \exists i\in\N, |v_{ij}| > \epsilon |\tmu_i|/(2n)\},$ and $S^{+} := S \setminus S^{-\eps}$. Allocate each item $j$ in $S^+$ to any agent $i$ in $\argmax_i v_{ij}$. Then consider the envy graph corresponding to this allocation of $\M\backslash S^{-\eps}$, and modify the allocation by eliminating all the cycles in the envy graph. Allocate all the items in $S^{-\eps}$ to a sink agent in the acyclic envy graph, and denoted it as $i^t.$ 

%\noindent
The following claim will be useful in proving the final allocation of the algorithm is $\gPO.$ 

\begin{claim}\label{clm:pos-mms-agent}
If $S^{-\eps}\neq \emptyset$ then there exists an $i \in \N$ such that $v_i(\M)>0$.
\end{claim}
% \vspace{-0.5cm}

\begin{proof}
Every agent with $v(\M)<0$ has $\mu\le-1,$ from Lemma \ref{lem:mu1}. The value of any chore in $\nbg$ for any such agent is at most $\nicefrac{\epsilon}{2n}\le \nicefrac{\eps|\mu|}{2n}\le \nicefrac{\eps|\tmu|}{2n},$ as from Claim \ref{clm:mu_tmu}, $|\tmu|\ge |\mu|.$ Hence, if $S^{-\eps}\neq\emptyset,$ then the agent who values any item in $S^{-\eps}$ more than $\eps\tmu/(2n)$ has $v(\M)>0.$   
\end{proof}

Finally, we show the maximum loss in value of each agent in the rounding process, which will be used to ensure that the algorithm returns an $\aeMMS$ allocation. 
\begin{lemma}\label{lem:rounding-loss}
In the rounding process, $i^t$ loses at most $\epsilon/2$ value and every other agent $i$ loses at most $\epsilon|\tmu_i|/2$ value.
\end{lemma}
% \vspace{-0.5cm}

\begin{proof}
Every agent except $i^t$, in the worst case, loses all her shared goods and gains all her shared chores in $S^+,$ and has no shared chores in $S^{-\eps},$ as she only gains from the rounding of items in $S^{-\eps}$. Her maximum loss from the items in $S^+$ is at most $(n-1)\cdot \eps\tmu/(2n)\le \eps\tmu/2,$ as $|S^+|\le |S|\le n-1,$ from Lemma \ref{lem:shared}. For agent $i^t,$ her loss from $S$ in the worst case is at most $(n-1)\cdot \eps/(2n)\le \eps,$ as each item in $S$ has absolute value at most $\eps/(2n)$ for her. 
\end{proof}

\subsection{$\PTAS$ for $\aMMS+\PO$}

Algorithm \ref{algo:nonid-mixed-plus} combines the ideas in the previous sections and finds an $\aeMMS$ allocation if an $\aMMS$ allocation exists, else returns an empty allocation to indicate that no $\aMMS$ allocation exists. The algorithm works as follows. First, it finds the approximate $\MMS$ values of all agents using the algorithms from Section \ref{sec:identical} and Appendix \ref{appendix:ptas-neg-id}, and classifies all items as $\bg$ or $\nbg$. Then among all allocations of $\bg$ items where the corresponding LP has a solution, if any, it finds the combined allocation of $\bg$ and $\nbg,$ called $\A,$ with the highest social welfare where $\nbg$ may be fractionally allocated. 

From constraint \ref{lp:ci} of the LP, $\A$ is $\aMMS,$ and as shown in Lemma \ref{lem:aemms}, its rounded allocation, denoted as $\A^r,$ is $\aeMMS.$ To ensure $\A^r$ is also $\gPO,$ for technical reasons we require the sum of absolute values of all agents to be at least $\alpha$ whenever there is at least one agent with $v_i(\M)>0.$ If $\A^r$ does not satisfy this condition, Algorithm \ref{algo:nonid-mixed-plus} ensures a stronger guarantee, namely at least one agent values her own bundle at least $1$, by modifying $\A^r$ as follows. Let $\N^+$ be the set of agents with $v_i(\M)>0$. Note that, for an $i\in \N^+$, $v_i(\M)=n$ but $v_i(A_i^r)<\alpha$, and hence there exists an agent $k\neq i$ such that $i$ values $k$'s bundle more than $1$. We consider two cases based on if $k$ is in $\N^+$ or not. If $k\notin \N^+$, then $k$'s $\MMS$ value is negative, thus even if we re-allocate her bundle to $i$ and give $k$ nothing (lines 18-19), her $\aeMMS$ guarantee is maintained. If no such $(i,k)$ pair is found, then we go to the other case, where $k$ has to be given something. For this (lines 21-22), we construct a graph on $\N^+$ where there is an edge from $i$ to $k$ if $v_i(A^r_i)<\alpha$ and $v_i(A^r_k)\ge 1$. This graph has to have a cycle (See proof of Claim \ref{clm:Ar-agent-1}), and swapping bundles along the cycle gives value more than $1$ to every agent along the cycle.

%If $\A^r$ does not satisfy this condition, then Algorithm \ref{algo:nonid-mixed-plus} modifies $\A^r$ as follows: Let $\N^+$ be the set of agents with $v_i(\M)>0$. Note that, since an agent outside $\N^+$ has negative MMS value, her MMS gurantee is maintained even she gets nothing. 
%If there are agents $i\in \N^+$ and $k\notin \N^+$ such that $v_i(A^r_k)\ge 1$, then it re-allocates $A^r_k$ to agent $i$. Note that, since agent $k$ has negative MMS, giving her nothing is maintains the MMS guarantee. After all such re-allocations, 

%First it finds all pairs of agents $i,k$ where $v_i(\M)>0,$ $v_k(\M)<0$ and $v_i(\A^r_k)\ge 1.$ It re-allocates $\A^r_k$ to $i.$ After all such swaps, it constructs a graph where there is an edge from agent $i$ to agent $k$ if $v_i(\A^r_i)<\alpha$ and $v_i(\A_k^r)\ge 1$. Bundles are swapped along cycles in this graph. checks if there is a cycle where each agent has less than $\alpha$ value for her own bundle and more than $1$ for the next. If such a cycle is found, $\A^r$ is again modified by giving every agent the bundle of her successor along any one cycle.  

We will prove the correctness of the algorithm in the remaining section. In what follows, we denote by $\A=[\A_1\cdots,\A_n]$ and $\A^r=[\A^r_1\cdots,\A^r_n]$ respectively the fractional allocation that is rounded after Line \ref{line:to-round} and its rounded allocation. First we show that the algorithm returns an $\aeMMS$ allocation if an $\aMMS$ allocation exists. 

\begin{algorithm}[t]
\caption{$\aeMMS+\gPO$ allocation of mixed items to non-identical agents} \label{algo:nonid-mixed-plus}
\DontPrintSemicolon
  \SetKwFunction{Define}{Define}
  \SetKwInOut{Input}{Input}\SetKwInOut{Output}{Output}
  \Input{Instance $\MMSins$, $\alpha\in(\epsilon,1]$, $\gamma>0$ and $\epsilon>0$}
  \Output{$\aeMMS+\gPO$ allocation or report $\aMMS$ allocation does not exist}
  \BlankLine
$\epsilon\gets \min\{\epsilon,\frac{\gamma\alpha}{(1+\gamma)}\},flag\gets false$\;
    For all $i\in\N,$ $\A^r_i\gets \emptyset$\tcp*{initialize $\A^r$ as the empty allocation}
    $\A\gets$ lowest social welfare allocation, i.e., give every item $j$ to agent $i$ with smallest $v_{ij}$\;
    For all $i,$ $\overline{\MMS_i}\gets (1-\frac{\epsilon}{2})$-$\MMS_i$ value of agent $i$\tcp*{use Algorithm from Section \ref{sec:identical}}
    Define $\bg$ and $\nbg$ according to Definition \ref{def:big-sml}\;
    \For {all allocations $\PB=[B_1,B_2\dotsc,B_n]$ of $\bg$\label{line-enum-big}}{
        Solve the LP (Equations \eqref{lp:small}-\eqref{lp:vars}) for allocating $\nbg$ items\;
     \If{LP has a solution}{
     $flag\gets true$\;
     $\PA\gets$ Allocation of $\nbg$ in optimal LP combined with $\PB$\;
     $\A\gets$ Allocation from $(\A,\PA)$ with higher welfare, i.e., $\sum_i v_i(\A_i)$\label{line:welfare}
     }
    }
    \If{$flag=true$\label{line:to-round}}{
    Make allocation graph of $\A$ acyclic using Lemma \ref{lem:allocation-acyclic}\;
    Round off $\A$ and obtain $\A^r$ by applying the rounding method from Section \ref{sec:LP}\label{line:aemms}\;
    \If{$\sum_i|v_i(\A^r_i)|<\alpha$, and $\exists i:v_i(\M)>0$\tcp*{technical steps to get $\gPO$}\label{Line:gPO-modifyAr}}{
    $\N^+\gets $ set of agents $i$ with $v_i(\M)>0$\;
    \eIf{$\exists i\in\N^+,\ k\notin\N^+:\ v_i(\A^r_k)\ge1$ \label{Line:neg-to-pos}}{
    modify $\A^r$ by giving $\A^r_k$ to $i,$ and giving $k$ nothing\tcp*{note: k has negative $\MMS$}}{
    Construct the following directed graph $G(V,E).$ $V=\N^+,$ directed edge $(i\rightarrow k)\in E$ if $v_i(\A^r_i)<\alpha$ and $v_i(\A^r_k)\ge1$\tcp*{$G$ has a cycle as for all $i\in V:$ $v_i(\A^r_i)<\alpha\le 1$}
    Swap bundles along any $1$ cycle in $G$ by giving every agent her successor's bundle \label{line:cycle-1-alpha}
    %\tcp*{after removing one cycle some $i$ has $v_i(\A^r_i)\ge 1\ge \alpha$}
    }
    }}
    \Return $\A^r$
\end{algorithm}

\noindent 

\begin{lemma}\label{lem:aemms}
If the LP has a solution for any partition of $\bg$, then $\A^r$ is an $\aeMMS$ allocation.
\end{lemma}
% \vspace{-0.5cm}

\begin{proof}
First, we argue for the allocation obtained after the rounding step on Line \ref{line:aemms}. Consider agent $i^t.$ Since $i^t$ corresponds to a sink node in the envy graph, from Claim \ref{clm:sinkval} and Lemmas \ref{lem:rounding-loss} and \ref{lem:mu1}, her value for her bundle in $\A^r$ is at least  $1-\epsilon/2 \ge 1-\epsilon\ge (\alpha-\epsilon)\mu_{i^t}$ if $\tmu_{i^t}\ge 0,$ and $-1-\epsilon/2 \ge -1-\epsilon\ge (1+\epsilon)\mu_{i^1}\ge \frac{1}{(\alpha-\epsilon)}\mu_{i^1}$ otherwise. Next, every agent $i$ except $i^t$, according to constraint \eqref{lp:ci} of the LP, receives a bundle of value at least $c_i$ from $\nbg$ in the fractional allocation of $\nbg$ corresponding to $\A$. Thus, for all $i\ne i^t,$ their value for their bundle in $\A^r$ is at least $v_i(B_i)+c_i-n\epsilon\cdot|\tmu_i|/(2n)\ge \min\{(\nicefrac{1}{\alpha})\tmu_i,\alpha\tmu_i\}-\epsilon\cdot|\tmu_i|/2,$ from Lemma \ref{lem:rounding-loss} and by definition of $c_i.$ Combined with Claim \ref{clm:mu_tmu}, when $\tmu_i\ge 0,$ this is at least $(\alpha-\epsilon/2)\tmu_i\ge (\alpha-\epsilon/2)(1-\epsilon/2)\mu_i\ge (\alpha-\epsilon)\mu_i.$ When $\tmu_i<0,$ this value is at least $\frac{1}{\alpha}\tmu_i+\epsilon\tmu_i/2$. As $(\frac{1}{\alpha}+\epsilon/2)\le \frac{1}{(\alpha-\epsilon/2)}$, and $\tmu_i<0$, along with Claim \ref{clm:mu_tmu}, $(\frac{1}{\alpha}+\epsilon/2)\tmu_i\ge \frac{1}{(\alpha-\epsilon/2)(1-\epsilon/2)}\mu_i\ge \frac{1}{(\alpha-\epsilon)}\mu_i.$

Any further modifications of the allocation can occur when (a) there is a pair of agents $i,k$ with $v_i(\A^r_k)>1$ and $v_k(\M)<0,$ or when (b) there is a cycle of agents with value at most $\alpha$ for their own bundle and value at least $1$ for the next. The only agents whose value decreases in these steps are those with $v_k(\M)<0,$ who after the swap receive no item. As $v_k(\M)<0,$ then $\mu_k<0,$ hence they receive at least $ \mu_k\ge \frac{1}{(\alpha-\eps)}\mu_k$ valued bundle. As no other agents lose, they still retain an $\aeMMS$ bundle.   
\end{proof}

Note that, the steps after rounding maintains the MMS guarantee. Since by construction, the LP has to be feasible whenever $\bg$ items are allocated as per an $\aMMS$ allocation, %showing the LP is feasible whenever an $\aMMS$ allocation exists, 
we get as a corollary,

% \vspace{-0.5em}
\begin{restatable}{corollary}{CorrAMMS}\label{corr:aemms-when-exists}
If an $\aMMS$ allocation exists, Algorithm \ref{algo:nonid-mixed-plus} returns an $\aeMMS$ allocation.
\end{restatable}

Finally, we show the approximate Pareto optimality of the $\aeMMS$ allocation returned. This is the more involved part. For this, we use the notion of \textit{social welfare} of any allocation, defined as the sum of values of all agents. Formally, for an allocation $A=[A_1\cdots,A_n]$ among $n$ agents, define its social welfare as $w(A)=\sum_i v_i(A_i).$

\begin{lemma}\label{lem:po-when-exists}
If an $\aMMS$ allocation exists, then $\A$ has the highest welfare among all the $\aMMS$ allocations of $\M$ among $\N$ obtained by allowing $\nbg$ to be fractionally allotted.
\end{lemma}
% \vspace{-0.5cm}

\begin{proof}
Let $\Set^B \subseteq \Pi_n(\bg)$ be the set of all partitions of $\bg$ corresponding to which there is a fractional $\aMMS$ allocation, or in other words, for which the LP has a solution. For every partition $\PB$ in $\Set^B,$ the objective function of the LP ensures that the allocation of $\nbg$ returned by the algorithm has the highest social welfare among all allocations that satisfy the LP constraints. Hence, among all $\aMMS$ allocations where the partition of $\bg$ is $\PB,$ the allocation returned, say $\PA,$ has the highest social welfare. Formally, let $\Set^{A,\PB}$ be the set of all $\aMMS$ allocations corresponding to the partition $\PB.$ Then $\PA\in \argmax_{A\in \Set^{A,\PB}}\sum_i w(A).$ 

Let the set of allocations $\PA,$ one corresponding to each partition $\PB\in\Set^B,$ be $\Set^{A}.$ From Line \ref{line:welfare} of the algorithm, $\A$ has the highest social welfare among all. Formally, $\A$ is in $\argmax_{A\in\Set^{A}}\{w(A)\}.$ Combining with the above characterization of the allocations in $\Set^A,$ we have,
$$\A \in \max_{\PB\in \Set^B}\argmax_{A\in\Set^{A,\PB}}\left\{w(A)\right\},$$
thus proving the lemma.
\end{proof}

Next, we prove two key properties (Lemmas \ref{lem:gpdom-amms} and \ref{lem:gpdom-welfare}) of any \textit{integral} allocation that $\gamma$-Pareto dominates $\A^r.$ Suppose $\A^*$ is such an allocation. %denoted by $\A^*.$
\begin{lemma}\label{lem:gpdom-amms}
%If allocation $\A^*$ $\gamma$-Pareto dominates $\A^r$, then 
$\A^*$ is an integral $\aMMS$ allocation. 
\end{lemma}
% \vspace{-0.5cm}

\begin{proof}
We will use the following relation between $\epsilon, \alpha$ and $\gamma.$ We have,
\begin{equation}\label{eqn:eps-alp-gam}
  \epsilon\le \frac{\gamma\alpha}{(1+\alpha)}\Rightarrow \epsilon\le \gamma(\alpha-\epsilon)\Rightarrow \gamma\ge \frac{\epsilon}{(\alpha-\epsilon)}\Rightarrow (1+\gamma)\ge \frac{\alpha}{(\alpha-\epsilon)}.  
\end{equation}
We know that $\A^r$ is an $\aeMMS$ allocation. Hence, agents $i$ with $\mu_i\ge 0$ get a bundle of value at least $(\alpha-\epsilon)\mu_i\ge 0$ in $\A^r,$ hence get a bundle of value at least $(1+\gamma)(\alpha-\epsilon)\mu_i$ in $\A^*.$ From equation \eqref{eqn:eps-alp-gam}, this is at least $\alpha\mu_i.$
Next, consider agents with $\mu_i<0.$ If they receive a bundle of positive value in $\A^r,$ then they also receive a positive valued, hence a bundle of value more than $\mu_i\ge \frac{1}{\alpha}\mu_i,$ in $\A^*.$ And if they get a negative valued bundle of value at least  $\frac{1}{(\alpha-\epsilon)}\mu_i$ in $\A^r,$ then they get a bundle of value at least $\frac{1}{(1+\gamma)(\alpha-\epsilon)}\mu_i,$ which from equation \eqref{eqn:eps-alp-gam} and the fact that $\mu_i<0,$ is at least $\frac{1}{\alpha}\mu_i.$ Hence, $\A^*$ is an integral $\aMMS$ allocation.
\end{proof}

The next property of $\A^*$ is that it will have higher social welfare than (fractional allocation) $\A.$ To prove this, we first prove two technical claims. 

\begin{claim}\label{clm:ArWelfareThanA}
$w(\A^r)\ge w(\A)-\epsilon$ if $S^{-\eps}\ne\emptyset,$ else $w(\A^r)\ge w(\A).$
\end{claim}
\begin{proof}
We consider each step in Algorithm \ref{algo:nonid-mixed-plus} that changes the allocation from $\A$ to $\A^r$, and see how it changes the social welfare. The first step is making the allocation graph of $\A$ acyclic. Here every agent's value, hence the social welfare also remains the same. The next step is the rounding process. Here, first the items in $S^+$ are allotted to the agents with the highest value for them, hence the sum of values of the items from $S^+$ in $\A^r$ is at least as much as that in $\A.$ As no other item's allocation changes, the social welfare from them remains the same. Hence this part only improves the social welfare. Then the envy graph cycle elimination only improves the value of every agent, hence does not reduce the social welfare. At this point, $w(\A^r)\ge w(\A).$ If $S^{-\eps}=\emptyset,$ the rounding process ends, hence this inequality holds, otherwise from Lemma \ref{lem:rounding-loss}, allocating the items from $S^{-\eps}$ reduces the sink agent's value, hence the social welfare, by at most $\epsilon,$ giving $w(\A^r)\ge w(\A)-\epsilon.$ Now we show the next part of the algorithm does not reduce the social welfare of $\A^r,$ hence these relations remain true, thus proving the claim. 

First consider the if statement on Line \ref{Line:neg-to-pos}. For agents $i,k$ when $k$'s bundle is given to $i,$ only the allocation of items in $\A^r_k$ changes and the allocation of $\M\backslash\A^r_k$ remains the same. Now $v_k(\A^r_k)<\alpha$ and $v_i(\A^r_k)>1,$ otherwise this step would not be executed. Hence the social welfare changes by $v_i(\A^r_k)-v_k(\A^r_k)>1-\alpha\ge 0$. Finally, suppose some bundles are swapped along some cycle by executing Line \ref{line:cycle-1-alpha}. Every agent had value at most $\alpha$ for their bundle, and received a bundle of value at least $1.$ Thus, the value of every agent in the cycle changes by at least $1-\alpha>0,$ and every other agent's bundle, hence its value, remains the same. Thus, the social welfare does not decrease in this step as well.
\end{proof}

\begin{claim}\label{clm:Ar-agent-1}
If there is an agent $i$ with $v_i(\M)>0,$ then there exists some agent $i'$ with $v_{i'}(\A^r_{i'})\ge \alpha.$ 
\end{claim}
% \vspace{-0.5cm}
\begin{proof}
If the claim is true for the allocation obtained after rounding $\A,$ then this is the allocation returned, hence we are done. Otherwise, the if condition of Line \ref{Line:gPO-modifyAr} is executed. If the condition of Line \ref{Line:gPO-modifyAr} is true, then as the allocation before this line was $\aeMMS,$ $v_i(\A^r_i)\ge 0,$ and after obtaining $k$'s bundle, $v_i(\A^r_i \cup \A^r_k)\ge 0+1=1$. 

Otherwise, for every agent $i$ in $\N^+$ we have $v(\M)=n,$ thus there is at least one agent $k\in \N^+$ such that $v_i(\A^r_k)\ge 1.$ The graph $G(V,E)$ has an edge $i\rightarrow k$ in this case. As $\sum_i |v_i(\A^r_i)|<\alpha,$ for every $i$ we have $v_i(\A^r_i)<\alpha\le 1.$ Thus for every edge $i\neq k.$ Hence, $G$ has at least one cycle. After swapping along a cycle, all agents along the cycle receive the bundle of their successor, hence have value at least $1$ for their own bundle. Thus, the final allocation has some agent with value at least $1\ge\alpha$ for her bundle.
\end{proof}

\begin{lemma}\label{lem:gpdom-welfare}
$\A^*$ has higher social welfare than $\A$. 
\end{lemma}
% \vspace{-0.5cm}

\begin{proof}
For cleaner exposition, we will denote $v_i(A_i)$ by $v_i(A)$ for any allocation $A.$

By definition of a $\gamma$-Pareto dominating allocation, the social welfare of allocation $\A^*$ is,
\begin{equation}\label{eq:wA-gpo}
    \begin{split}
        w(\A^*) &> (1+\gamma)\cdot \sum\limits_{i:v_i(\A^r)\ge 0}v_i(\A^r) + \frac{1}{1+\gamma}\cdot \sum\limits_{i:v_i(\A^r)< 0}v_i(\A^r)\\
        &= \sum_i v_i(\A^r) + \gamma\cdot\sum\limits_{i:v_i(\A^r)\ge 0}v_i(\A^r) - \frac{\gamma}{(1+\gamma)}\cdot\sum\limits_{i:v_i(\A^r)< 0}v_i(\A^r)\\
        &=w(\A^r)+ \gamma\cdot\sum\limits_{i:v_i(\A^r)\ge 0}v_i(\A^r) + \frac{\gamma}{(1+\gamma)}\cdot\sum\limits_{i:v_i(\A^r)< 0}|v_i(\A^r)|. 
    \end{split}
\end{equation}
Let $\mathds{1}_{S^{-\eps}}$ be an indicator variable for whether $S^{-\eps}$ is non-empty. Substituting the relation between $w(\A)$ and $w(\A^r)$ from Claim \ref{clm:ArWelfareThanA} in equation \eqref{eq:wA-gpo} we get,
$w(\A^*)>w(\A)-\epsilon\cdot \mathds{1}_{S^{-\eps}}+ \gamma\cdot \sum_{i:v_i(\A^r)\ge 0}v_i(\A^r) + \frac{\gamma}{(1+\gamma)}\cdot\sum_{i:v_i(\A^r)< 0}|v_i(\A^r)|.$ Hence, to prove the lemma, it suffices to show, 
\begin{equation}\label{eq:gpo-proof}
  \gamma\cdot\sum_{i:v_i(\A^r)\ge 0}v_i(\A^r) + \frac{\gamma}{(1+\gamma)}\cdot\sum_{i:v_i(\A^r)< 0}|v_i(\A^r)|-\epsilon\cdot \mathds{1}_{S^{-\eps}}\ge 0.
\end{equation}
If $S^{-\eps}=\emptyset,$ we are done, as all values on the left hand side in the above equation, say L, are non-negative. Hence, we now prove the equation when $S^{-\eps}\ne\emptyset,$ that is, $\mathds{1}_{S^{-\eps}}=1.$

We know $\epsilon\le \frac{\gamma\alpha}{(1+\gamma)},$ and as $\gamma\ge 0,$ $\gamma\ge \frac{\gamma}{1+\gamma}.$ Substituting these relations in L, we have, $$L\ge \frac{\gamma}{1+\gamma}\cdot\sum_{i:v_i(\A^r)\ge 0}v_i(\A^r) + \frac{\gamma}{(1+\gamma)}\cdot\sum_{i:v_i(\A^r)< 0}|v_i(\A^r)|-\frac{\gamma\alpha}{(1+\gamma)}=\frac{\gamma}{1+\gamma}\left(\sum_i |v_i(\A^r)|-\alpha \right).$$  But as $S^{-\eps}\ne \emptyset,$ from Claim \ref{clm:pos-mms-agent}, there is some agent $i$ with $v_i(\M)>0.$ Then from Claim \ref{clm:Ar-agent-1}, $v_{i'}(\A^r)\ge \alpha$ for at least one agent $i'$. Hence, $\frac{\gamma}{1+\gamma}\left(\sum_i |v_i(\A^r)|-\alpha \right)\ge 0,$ thus proving equation \eqref{eq:gpo-proof}, hence the lemma.  
\end{proof}

\begin{corollary}\label{corr:gpo-when-exists}
If an $\aMMS$ allocation exists, Algorithm \ref{algo:nonid-mixed-plus} returns a $\gPO$ allocation.
\end{corollary}
% \vspace{-0.5cm}

\begin{proof}
For contradiction, suppose $\A^r$ is not $\gPO.$ Then there is another allocation, say $\A^*,$ that $\gamma$-Pareto dominates $\A^r.$ From Lemmas \ref{lem:gpdom-amms} and \ref{lem:gpdom-welfare}, $\A^*$ is an integral $\aMMS$ allocation with higher social welfare than $\A.$ From Lemma \ref{lem:po-when-exists}, this is a contradiction.\end{proof}

Using Corollaries \ref{corr:aemms-when-exists} and \ref{corr:gpo-when-exists}, the next theorem obtains the main result. %, and runtime analysis from Lemma \ref{lem:runtime-amms}, we have the main result. %together show the following theorem. 
\begin{restatable}{thm}{thmNonId}\label{thm:non-identical}
Given an instance $\MMSins$ and constants $\alpha,\epsilon,\gamma>0,$ that is, an instance of the $\aMMS+\PO$ problem, Algorithm \ref{algo:nonid-mixed-plus} returns an $\aeMMS+\gPO$ allocation if an $\aMMS$ allocation exists, else reports it does not exist, in time $O(2^{O((n^3\log n)/\min\{\epsilon^2,\gamma^2/2\})}m^3)$ where $n=|\N|$ is a constant. Thus, it is a $\PTAS$. % for the $\aMMS+\PO$ problem. 
\end{restatable}
%\vspace{-0.5em}
%\begin{restatable}{lemma}{lemRuntime}\label{lem:runtime-amms}
%Algorithm \ref{algo:nonid-mixed-plus} runs in time $O(m^3).$
%\end{restatable}

\subsection{$\PTAS$ for $\optaMMS+\PO$}
Our final goal is to solve $\optaMMS+\PO$ problem, that is to find $\aMMS+\PO$ allocation for the highest possible $\alpha$. In this section we design a $\PTAS$ for this problem: given $\MMSins$, and constants $\epsilon,\gamma>0$, we design a polynomial-time algorithm to find $\aeMMS+\gPO$ allocation for the highest possible $\alpha$. For this we will use the $\PTAS$ for $\aMMS+\PO$ problem described in the previous section. 

Note that, for a given $\alpha$, Algorithm \ref{algo:nonid-mixed-plus} either returns $\aeMMS+\gPO$ allocation, or returns an {\em empty allocation}. And by Theorem \ref{thm:non-identical}, whenever it returns an empty allocation no $\aMMS$ allocation exists. Using this, we run a simple binary search to find the highest value of $\alpha\in [\epsilon,1]$ % within a constant additive error $\gamma>0$ 
(up to a polynomial precision)
for which Algorithm \ref{algo:nonid-mixed-plus} returns a non-empty allocation. If empty allocation is returned for every $\alpha$ in our search, then we only need to ensure $\PO$ and therefore we return the social welfare maximizing allocation obtained by giving every item to the agent who values it the most. 

We stop when the range of $\alpha$ values under consideration is $[c-\delta, c+\delta]$ for $\delta=\frac{1}{2^{poly(m)}}$ for some constant $c$, where $m=|\M|$. Clearly the number of iterations the binary search will take to get within such a range is at most $poly(m)$. Each iteration runs Algorithm \ref{algo:nonid-mixed-plus} once, and hence finishes in $O(m^3)$ time (Theorem \ref{thm:non-identical}). Thus the overall running time of the algorithm is $poly(m)$, and the next theorem follows.
\begin{restatable}{thm}{ThmNonidmain}\label{thm:nonid-optammspo}
Given an instance $\MMSins$ and constants $\epsilon,\gamma>0,$ that is, an instance of the $\optaMMS+\PO$ problem, there is a $\PTAS$ that runs for $(2^{O(1/\min\{\epsilon^2,\gamma^2\})}poly(m))$ time and returns an $\aeMMS+\gPO$ allocation such that for any $\alpha'>\alpha+\frac{1}{2^{m^c}},$ no $\alpha'$-$\MMS$ allocation exists, where $c>0$ is a constant.
\end{restatable}
%The number of iterations to reduce the range of possible $\alpha$ values to $[c-\gamma, c+\gamma]$ for some constant $c$ is at most $\log (1/\gamma).$ Each iteration runs Algorithm \ref{algo:nonid-mixed-plus}. Hence, from Lemma \ref{lem:runtime-amms}, the total time to find the required allocation is $O(\log (1/\gamma)m^3)=O(m^3)$ as $\gamma$ is a constant. 
This completes the discussion of the $\optaMMS+\PO$ problem. Next, we describe a $\PTAS$ for finding the $\MMS$ value of an agent for distributing a set of items $\M$ into $n$ bundles according to a valuation function $v:\M\rightarrow \mathbb{R}$. We refer to this problem as the $\aMMS$ problem with ($n$) identical agents. Using Lemma \ref{lem:mms-sign} we can find the sign of the $\MMS$ value. Hence, we describe two algorithms, one for each case when $\MMS\ge 0$ and otherwise. The following section discusses the algorithm for the former case.

\section{Finding $\MMS$ values of agents when $\MMS\ge 0$}\label{sec:identical}

%Note that this problem is equivalent to the $\aMMS$ problem when all the agents have valuations identical to the agent concerned. Hence, for this $\aMMS$ problem, $\alpha=1,$ and the agents have the same $\MMS$ values.  In this section, we describe a $\PTAS$ for the case with non-negative $\MMS$ values. The $\PTAS$ for the other case is similar, and is discussed in Section \ref{appendix:ptas-neg-id}.

%\subsection{Computing $\MMS$ when $\MMS\ge 0$}

In this section we prove Theorem \ref{thm:identical} for the case when $\MMS\ge 0$, i.e., given an instance $\MMSidins,$ we describe an algorithm to find a $(1-\epsilon)$-$\MMS$ allocation for any constant $\epsilon>0$. 
Using scale invariance (Lemma \ref{lem:scale}), here on we assume $v(\M)=n$ without loss of generality. Due to Lemma \ref{lem:avg}, this implies $\MMS\le v(\M)/n=1$. 

The high level ideas used in the algorithm are as follows. First is a classification of all the items into two sets, $\Bigi$ and $\Sml,$ based their value (Section \ref{sec:id-big-small}). Using this we prove that $|\Bigi|$ is constant, which allows the enumeration of all allocations of $\Bigi,$ referred as \textit{partitions} of $\Bigi$ to avoid confusion with allocations of $\M.$ Next in Section \ref{sec:id-val-inval} we explain a short procedure which allows us to characterize partitions of $\Bigi$ as \textit{valid} or \textit{invalid}. We show that there is at least one valid partition corresponding to which there is an $\MMS$ allocation, hence all invalid partitions can be discarded. Finally in Section \ref{sec:bag-fill}, we describe a sub-routine called Bag-Fill, that greedily allocates or `fills' the items from $\Sml$ upon partitions or `bags' of $\Bigi$ that satisfy certain constraints to obtain a $(1-\epsilon)$-$\MMS$ allocation. The main algorithm (Section \ref{sec:id-ptas-pos}) enumerates all partitions of $\Bigi,$ discards the invalid partitions, and applies Bag-Fill if its constraints are satisfied. If the constraints are not satisfied, we show that we can apply the $\PTAS$ for obtaining $\MMS$ allocations with identical agents for a goods manna, and obtain a $(1-\eps)$-$\MMS$ allocation. 

We now discuss these key ideas formally followed by the algorithm in separate subsections.
% \vspace{-1em}
\subsection{Big and Small items}\label{sec:id-big-small}
Given an instance $\MMSidins$ and a constant $\epsilon>0$, let $\Bigi$ be the set of items in $\M$ which have absolute value higher than $\epst$, i.e., $$ \Bigi = \{ j \in \M: |v_{j}| \ge \epst \}. $$ Let $\Bigi^+$ and $\Bigi^-$ respectively be the sets of the goods and the chores in $\Bigi$, i.e., $\Bigi^+=\Bigi\cap \M^+$, and $\Bigi^-=\Bigi\cap \M^-$. Let $\Sml$ be the set of small items, i.e., $\M \setminus \Bigi,$ and similarly define the sets of goods $(\Sml^+)$ and chores $(\Sml^-)$ in $\Sml$. 

We abuse notation slightly and call items in the set $\Bigi$ (or $\Sml$) as $\Bigi$ (resp. $\Sml$) items. 
\begin{lemma}\label{lem:big-const}
$|\Bigi|=O(n/ \epsilon).$
\end{lemma}
\begin{proof}
As $v(\M)=n$, we have $n = v(\M^+) - |v(\M^-)|.$ Then, as $v(\M^+)\ge (1+\tau)|v(\M^-)|,$
\begin{equation}
    \label{eq:bound-on-M^+}
     n \ge v(\M^+) - \frac{v(\M^+)}{1+\tau} \implies v(\M^+) \le \frac{n(1+\tau)}{\tau} \ .
\end{equation}
Finally, by the definition of $\Bigi^+$ we have \begin{equation*}
    \label{eq:iden2}
    |\Bigi^+| \le \frac{v(\M^+)}{\epst} \le  \frac{2n(1+\tau)}{\epsilon\tau}  \ .
\end{equation*}
Similarly, we have 
\begin{equation}
\label{eq:bound-on-M^-}
n\ge (1+\tau)|v(\M^-)|-|v(\M^-)|\Rightarrow |v(\M^-)| \le n/\tau.    
\end{equation}
Thus, the number of $\Bigi$ chores is bounded as $|\Bigi^-| \le  \frac{2n}{\epsilon\tau}  $. Hence, $|\Bigi|=|\Bigi^+\cup\Bigi^-|=O(\epsilon)$.
\end{proof}

As $n$ and $\epsilon$ are constant, Lemma \ref{lem:big-const} implies that all partitions of $\Bigi$ can be enumerated in constant time.
%\begin{corollary} \label{cor:constant-pars}
%The set of all partitions of $\Bigi$, namely $\Pi_n(\Bigi)$, can be enumerated in constant time. 
%\end{corollary}

\subsection{Valid and Invalid partitions of $\Bigi$}\label{sec:id-val-inval}

Given an allocation $\PA=[A_1\cdots,A_n]$ of $\M,$ denote by $\PB=[B_1\cdots, B_n]$ the allocation from $\M\backslash \Sml^+.$ Classify the bundles of $\PA$ based on their value from $\PB$ into sets $\B_1$, $\B_2$, $\B_3$, and $\B_4,$ together called the $\B$-sets, as follows.
\begin{equation}\label{eq:b1-4}
\begin{array}{llll}
    \B_1 &:= \{A_k \in \PB: v(B_k) > 1 \} &
    \B_2 &:= \{A_k \in \PB:1-\epsilon \le v(B_k) \le 1 \} \\
    \B_3 &:= \{A_k \in \PB: 0 \le v(B_k) < 1-\epsilon \} &
    \B_4 &:= \{A_k \in \PB:   v(B_k) < 0 \} 
\end{array} 
\end{equation}

We will abuse notation to denote all items in the sets in any $\B_i,$ i.e. $\cup_{\A_k \in \B_i} \A_k,$ by $\B_i.$

Given a partition of $\Bigi$ $\PB,$ we now explain a procedure using which we classify the partition as valid or invalid. First classify the bundles of $\PB$ into four sets as per equation \eqref{eq:b1-4}. Initially, all $\Sml$ items are unallocated. Then while $\B_1\ne \emptyset$ and $\Sml^-\ne \emptyset,$ assign any item from $\Sml^-$ to any agent that has a bundle from $\B_1.$ Re-classify the bundles using equation \eqref{eq:b1-4} and remove the assigned item from $\Sml^-$ after every assignment. This procedure ends when either $\B_1$ or $\Sml^-$ becomes empty, or both. If in the end $\B_1\ne \emptyset$, and we also have (a) $\B_4\ne \emptyset$ and (b) $v(\B_3\cup\B_4\cup\Sml^+)<(1-\epst)(|\B_3|+|\B_4|)$ then we call $\PB$ \textit{invalid}. All partitions of $\Bigi$ that are not invalid are called valid.  

\begin{lemma}
\label{lem:continue}
There exists an $\MMS$ allocation where the $\Bigi$ items are allocated according to a valid partition.
\end{lemma}
\begin{proof}
Let $\PA$ be an $\MMS$ allocation with the lowest value of $|\B_1|+|\B_4|,$ and let $\PB$ be its corresponding partition of $\Bigi$. Suppose $\PB$ is invalid. Then the agents with bundles from $\B_3$ or $\B_4$ can only receive items from $\B_3\cup \B_4\cup\Sml^+$ in $\PA$. From Lemma \ref{lem:avg} and the definition of invalid partitions giving $v(\B_3 \cup \B_4 \cup \Sml^+) < (1-\epst)(|B_3| + |B_4|),$ we have $\mu < (1-\epst)$. 

As both $\B_1,\B_4\ne \emptyset,$ consider any bundles $B$ and $B'$ respectively from $\B_1$ and $\B_4$. Let $v(B) = (1+x)$ and $v(B')=-y$ for some $x,y >0$. In $\PA$, $B'$ must be bundled with a set of items from $\Sml^+$, denoted as $\Set^+,$ of value at least $y+\mu$. We form two new bundles $A_1$ and $A_2$ of at least $\MMS$ value using $B,$ $B'$ and $\Set^+$ as follows. First merge $B$ and $B'$ into one bundle (with value $1+x-y$). If $1+x-y\ge\mu$, call this bundle $A_1$ and add all the remaining items from $\Set^+$ to $A_2$. Each bundle thus has value at least $\mu$. Otherwise, if $1+x-y < \mu$, add items from $\Set^+$ one by one, each time to the bundle with the lower value before adding the item. Let $A_1$ and $A_2$ be the resulting bundles after adding all items in $\Set^+$. Without loss of generality, let $v(A_1) \ge v(A_2)$. %The second inequality holds because 
As each item in $\Set^+$ has value at most $\epst$ and is always added to the lower valued bundle, $v(A_1) - v(A_2) \le \epst$. Thus, 
\begin{align*}
    v(A_1) + v(A_2) \ge 1+x-y+y+\mu=  1+x+\mu \text{ and }v(A_1) - v(A_2) \le \epst\\
    \implies v(A_2) \ge \tfrac{1}{2}(1+x+\mu-\epst) >  \tfrac{1}{2}(1+\mu-\epst) > \mu \implies v(A_1) > \mu.
\end{align*} 
 %Unless $v(B) +v(B')$ is much bigger than $x+\mu^*$ which in this case the last inequality still holds ($v(A_i) > \mu^*$).

No item from $\Bigi$ is assigned to $A_2.$ Thus, $A_2\notin \B_1\cup\B_4,$ and % and depending on if $1+x-y\ge \mu,$ either $A_1\notin \B_1$ or $A_1\notin \B_4.$ 
$A_1$ and $A_2$ combined with the allocations of the remaining agents who did not get $B$ or $B'$ in $\PA$ form an $\MMS$ allocation with a smaller value of $|\B_1|+|\B_4|$ than $\PA,$ a contradiction. Thus, $\PB$ is valid.
\end{proof}

\subsection{Algorithm Bag-Fill}\label{sec:bag-fill} In this section we design the algorithm Bag-Fill (Algorithm \ref{algo:Bag-Filling}) that generalizes algorithms in \cite{ghodsi2018fair,garg2018approximating,garg2019improved} to the mixed setting. Bag-Fill (Algorithm \ref{algo:Bag-Filling}) takes as input an $\MMS$ instance $\MMSidins$, and a partition of the $\Bigi$ items of $\M$, denoted by $\PB=[B_1,B_2,\dots,B_n]$ such that they satisfy one of the two condition sets \eqref{eq:cond-set-1} or \eqref{eq:cond-set-2}. It outputs an allocation of items $\PA=[A_1,\dots,A_n]$ where $v(A_i) \ge 1-\epsilon$, for all $i\in [n]$.
% \vspace{-2.5em}
\begin{multicols}{2}
\noindent
\begin{equation}\label{eq:cond-set-1}
    \begin{split}
    & v(\Sml) + \sum_{k=1}^n v(B_k) \ge n. \\
    & v(B_k) \le 1 \ \ \forall k \in [n] . \\
    & |v_j| < \epsilon  \ \ \forall j \in \Sml.
    \end{split}
\end{equation}
\begin{equation}\label{eq:cond-set-2}
    \begin{split}
    & v(\Sml) + \sum_{k=1}^n v(B_k) \ge n (1-\epst).\\
    & v(B_k) \le 1-\epst \ \ \forall k \in [n] .\\
    & |v_j| <\epst  \ \ \forall j \in \Sml.     
    \end{split}
\end{equation}
\end{multicols}
%Given this input, Algorithm \ref{algo:Bag-Filling} outputs an allocation of items $A=\{A_1,\dots,A_n \}$ where $v(A_i) \ge 1-\epsilon$, for all $i\in [n]$. 
Algorithm \ref{algo:Bag-Filling} works as follows. It has $n-1$ rounds. Each round starts with a bundle (`bag') from $\PB$. If the bag is valued at least $(1-\epsilon)$, then it is assigned to some agent. If not, we first add all the unallocated $\Sml$ chores to this bag. Then one by one we add the unallocated goods from $\Sml$ until it is valued at least $(1-\epsilon),$ and assign to some agent. After all rounds are done, in the last step, all remaining items from $\Sml$ are added to the bag $B_n$. The next lemma proves the correctness of the algorithm. 

\begin{algorithm}[tbh!]
\caption{Bag-filling to find $(1-\epsilon)$-$\MMS$ allocation of identical agents} \label{algo:Bag-Filling}
\DontPrintSemicolon
  \SetKwFunction{Define}{Define}
  \SetKwInOut{Input}{Input}\SetKwInOut{Output}{Output}
  \Input{$\MMSidins$, Partition of $\Bigi\subseteq \M$: $\PB=[B_1,B_2,\dots,B_n].$ Input satisfies Condition set \eqref{eq:cond-set-1} or \eqref{eq:cond-set-2}}
  
  \Output{$\PA = \{A_1,\dots,A_n \}$ such that $v(A_i) \ge 1-\epsilon,\ \forall i\in [n].$}
  \BlankLine
  $\PA \gets \emptyset $, $\Sml\gets \M\backslash \Bigi$\;
  \For{$k \in \{1,\dots,n-1 \}$} {
  \While{$v(B_k) < 1- \epsilon$}{
  $j \gets \argmin_{j \in \Sml} v_j$\;
  $B_k \gets B_k \cup \{j\}, \Sml \gets \Sml \setminus \{j\} $\;
  }
  $\PA \gets \PA \cup \{B_k \}$\;
   }
   $B_n \gets B_n \cup \Sml,\ \ \PA \gets \PA \cup \{B_n \}$\;
  \Return $A$

\end{algorithm}
\begin{lemma}
\label{lem:bag-filling}
If an $\MMS$ instance with identical agents satisfies condition set \eqref{eq:cond-set-1} or \eqref{eq:cond-set-2}, then Algorithm \ref{algo:Bag-Filling} gives a $(1-\epsilon)$-Allocation.
\end{lemma}
\begin{proof}
By induction on $k\in\{0,1,2,\cdots,n-1\}$, we prove that the value of each assigned bundle after $k$ rounds is in $[1-\epsilon, 1]$ if the instance satisfies condition set \eqref{eq:cond-set-1}, and in $[1-\epsilon,1-\epsilon/2]$ if it satisfies condition set \eqref{eq:cond-set-2}. The base case when $k=0$ is trivial.

First consider the case when condition set \eqref{eq:cond-set-1} is satisfied. Assume the value of all bundles assigned to the first $k-1$ agents are in this range. Now $v(B_j)\le 1$ for all $j\in \{k,..,n\}$. If $v(B_k)\ge 1-\epsilon,$ we are done. If not, then while the value of $B_k$ is less than $1-\epsilon,$ the value of the unallocated items from $\Sml$ is at least the value of all items minus that of all the allocated bundles and unallocated bags of $\Bigi$ items. This can be bounded as, 
$$v(\M)-\sum_{i< k} v(A_i)-v(B_k)-\sum_{i>k} v(B_i) > n - (k-1)-  (1-\epsilon) -(n-k)>\epsilon.$$ 
Hence, there is at least one unallocated $\Sml$ good. Before adding the last $\Sml$ good to $A_k$, its value was strictly less than $1-\epsilon$. Adding the last item increases the value by at most $\epsilon$. Hence, the value of $A_k$ is at most $1$. Thus, $v(A_k)\in [1-\epsilon,1],$ for all $k\in[n-1].$ As the total value of all items is at least $n$, and the total value of the $n-1$ assigned bundles is at most $n-1$, the last agent also gets a bundle of value at least $1$. 

Now suppose the instance satisfies condition set \eqref{eq:cond-set-2}. In every round, while this is not true, the value of unallocated goods is at least $\epsilon/2.$ Thus, there is at least one $\Sml$ good. Finally, after assigning $n-1$ bundles, the total value remaining is at least $n(1-\epsilon/2)-(n-1)(1-\epsilon/2),$ hence the last bundle also has value at least $(1-\epsilon)$.
\end{proof}

% In proof of Lemma \ref{lem:bag-filling}, there is a slack of $\eps$ in some bags but we keep that to save the generality of Algorithm \ref{algo:Bag-Filling} to be used in all steps of Algorithm \ref{algo:MMS-value}.

%The second condition set can also be proved to work using the similar steps.

%In this section, we describe Algorithm \ref{algo:MMS-value}, which gives a $(1-\epsilon)$ $\MMS$ allocation $A = \{A_1,\dots,A_n \}$ for a proper instance $\I= \langle N,M,V \rangle$ where agents in $N$ are all identical. Let $n$ be the number of agents, $X$ be the set of unallocated items where initially $X = M$, and $V \in \mathbb{R}^{|M|}$ be the valuation function of all agents.

%\subsubsection*{Algorithm}

\subsection{The $\PTAS$}\label{sec:id-ptas-pos} 
We use the notions from the previous subsections and derive the $\PTAS,$ shown in Algorithm \ref{algo:MMS-value}. The $\PTAS$ works as follows. It first enumerates all the partitions of $\Bigi$. For each partition $\PB$, it first classifies the bundles into the $\B$-sets as per equation \eqref{eq:b1-4}. If $\B_1$ is not empty, then add items from $\Sml^-$ to any bag in $\B_1,$ re-defining the sets and removing the assigned item from $\Sml^-$ after each assignment. This process ends when either $\B_1=\emptyset$ or $\Sml^-=\emptyset$. In the first case, condition set \eqref{eq:cond-set-1} of the Bag-Fill algorithm is satisfied, and we run Algorithm \ref{algo:Bag-Filling} (Line \ref{line:bf2}) and return its output. 

Otherwise when $\B_1\ne \emptyset$, if $|\B_4|=0$, then reduce to the following goods manna $\aMMS$ problem instance $(\N',\M',v')$. $\N'$ is the set of agents who received bundles from $\B_3$ or $\B_4$. $\M'$ has (a) $\Sml^+,$ with each item having the same value in $v'$ as in $v$, and (b) for each bundle $B\in \B_3,$ $\M'$ has a new item $b$ with value $v'_b = v(B)$. 
Run the $\PTAS$ from \cite{Jansen2016} on $(\N',\M',v')$ to find a $(1-\epsilon)$-$\MMS$ allocation of $\M'$ among the $n'=n-(|\B_1|+|\B_2|)$ agents, and store its output in $\ASet$. 

For the final case when $\B_1\ne \emptyset$ and $|\B_4|>0$, first check if the remaining unallocated goods and agents in $\B_3 \cup \B_4$ fulfill the condition set \ref{eq:cond-set-2}. If they do, apply the Bag-Fill and return the $(1-\epsilon)$-$\MMS$ allocation. If not, then $\PB$ is invalid, hence discarded. After enumerating all the partitions, the algorithm returns the best allocation from $\ASet$.

\begin{algorithm}[tbh!]
\caption{$(1-\epsilon)$-$\MMS$ allocation for identical agents with $\MMS\ge 0$} \label{algo:MMS-value}
\DontPrintSemicolon
  \SetKwFunction{Define}{Define}
  \SetKwInOut{Input}{Input}\SetKwInOut{Output}{Output}
  \Input{$\MMSidins$ such that $v(\M)=n$, $\epsilon\in[0,1]$}
  \Output{$(1-\epsilon)$-$\MMS$ Allocation}
  \BlankLine
  $\ASet \gets \emptyset.$ $\Pi_n(\Bigi)\gets $ all partitions of $\Bigi$ into $n$ sets. \label{line:initial} \;
  \For{$\PB \in \Pi_n (\Bigi)$ \label{line:enumerate}} {
  Define $\B$-sets as per equation \eqref{eq:b1-4}. \label{line:empty-A}\;
  \While {$\B_1 \ne \emptyset$ and $\Sml^-\ne \emptyset$ \label{line:coverB1}} {
  Remove any item $j$ from $\Sml^-$ and assign to any agent with a $\B_1$ bundle\;
  Re-define the $\B$-sets for the new allocation
}
\If{$\B_1 = \emptyset$}{
$\PA \gets $ Bag-Fill($(\N,\M,v)$, $\B$-sets) \label{line:bf2} \;
  \Return $\PA$\;
  }
$\A^{1,2}\gets$ allocation of all bundles from $\B_1$ and $\B_2$ to distinct agents\;
$\N'\gets$ set of remaining agents, $\M' \gets \cup_{B\in\B_3\cup\B_4}B \cup \Sml^+$, $n'=|\N'|$\;
  \eIf{$\B_4\ne \emptyset$}{
   \uIf{$v(\M') \ge n' (1- \epst) $}{ \label{line:enough}
  $\PA \gets \A^{1,2} \ \cup$ Bag-Fill($(\N',\M',v),\B-sets=\B_3\cup\B_4$) \label{line:bf3}\;
  \Return $\PA$ \label{line:return-bf1}\;
  }
  \Else{
  \textbf{continue}\tcp*{$\PB$ is invalid}
  }
  }{
  $\M'\gets \Sml^+$\;
%  $\N'\gets$ agents who received $\B_3\cup \B_4$,
%  $\M'\gets \Sml^+$, $\forall j\in \M',\ v'_j=v_j$.\;
  \For {$B \in \B_3$}{
  introduce a new good $b$ with $v(b)=v(B);$ $\M' \gets \M' \cup \{b \}$}
  $\PA \gets \A^{1,2} \cup (1-\epsilon)$-$\MMS$ allocation for $(\N',\M',v')$ using the algorithm in \cite{Jansen2016}\label{line:good-only}\;    $\ASet \gets \ASet \cup \{\PA\}$ \label{line:store}
    }
  }
  \Return $\argmax_{\PA\in \ASet} \min_{A_i \in \PA} v(A_i)$ \label{line:last}
\end{algorithm}

Let us now discuss the analysis of the $\PTAS.$ To prove correctness when $\B_1\ne \emptyset$ and $\B_4=\emptyset$, we first show in Lemma \ref{lem:valid-reduction} a relation between the $\MMS$ values of the given instance and the reduced goods manna instance. Let $\PA^*$ be some $\MMS$ allocation, and $\PB^*=\{B^*_1,B^*_2,\cdots,B^*_n\}$ be the allocation of $\Bigi$ items according to $\PA^*$.

\begin{lemma} \label{lem:valid-reduction}
If for $\PB^*,$ the subsequent allocation of $\Bigi\cup\Sml^-$ in Algorithm \ref{algo:MMS-value} has $\B_4= \emptyset,$ then, 
$$\MMS^n(\M) \le_p  \MMS^{n- |\B_1|-|\B_2|} (\bigcup_{B \in \B_3 } B \cup \Sml^+)\ . $$
\end{lemma}
\begin{proof}
%First, as we stop adding $\Sml^-$ items to bundles from $\B_1$ if their value goes below $1,$ every bundle originally categorized
We form an allocation of $\cup_{B\in\B_3} B\cup\Sml^+$ among $n-|\B_1|-|\B_2|$ agents with the smallest bundle's value at least $\MMS^n(\M),$ thus proving the lemma. Consider the allocation of $\Sml$ in $\PA^*$. Allocate the items from $\M'=\cup_{B\in\B_3}B\cup \Sml^+$ among the set $\N'$ of agents who have received bundles in $\B_3$, as they are allocated in $\PA^*$. Call this allocation $\PA'$. Now the allocation $\PA^*$ may also have some $\Sml$ chores assigned to agents in $\N'$, but no other goods. The lowest valued bundle in $\PA'$ thus has value at least that of the lowest valued bundle in $\PA^*$ (since no $\nbg$ chore is added to these bundles in $\PA'$). The $\MMS$ value of agents in $\N'$, when partitioning $\M'$ among them, is at least that of the lowest valued bundle of $\PA'$, hence is at least $\MMS^n(\M)$.
\end{proof} 

Next we state and prove the main theorem of this section.

\begin{theorem}\label{thm:main-id-pos}
Given an instance $\MMSidins$ with $\MMS\ge 0$, Algorithm \ref{algo:MMS-value} returns a $(1-\epsilon)$-$\MMS$ allocation in $O(m)$ time.
\end{theorem}
\begin{proof}
First we prove the correctness of the algorithm. Note that no valid partition is discarded, as the procedure before deciding to discard a partition is exactly the procedure to determine if the partition is invalid. Consider a valid partition $\PB^*$ corresponding to an $\MMS$ allocation $\PA^*$, and its $\B$-sets as per \eqref{eq:b1-4}. From Lemma \ref{lem:continue}, such a partition exists. After executing the while loop on Line \ref{line:coverB1}, as every $\Sml$ chore has absolute value at most $\epsilon/2,$ upon adding the last chore before the value falls below $1,$ the value of every bundle to which a chore was added is still at least $1-\epsilon/2.$ After this, one of the cases based on which conditions from $\B_1=\emptyset$ and $\B_4=\emptyset$ are true gets executed. In every case, there is some allocation generated, as the partition is valid. 

If the Bag-Fill algorithm is called, then every agent gets a bundle of value $1-\epsilon.$ As $\MMS\le 1,$ the allocation returned is $(1-\epsilon)$-$\MMS.$

If the $\PTAS$ of \cite{Jansen2016} is called, then first, the agents receiving bundles from $\B_1,$ $\B_2,$ by definition of these sets, have value at least $1-\eps\ge (1-\epsilon)$-$\MMS$ for their bundle. Also, as $\PB^*$ corresponds to an $\MMS$ allocation, the $\MMS$ value for allocating the remaining items among the remaining agents, from Lemma \ref{lem:valid-reduction}, is at least the original $\MMS$ value. Hence, a $(1-\eps)$-$\MMS$ allocation of the goods manna instance, combined with the allocations to the agents with the $\B_1$ and $\B_2$ bundles, is $(1-\epsilon)$-$\MMS.$ 

As $\PB^*$ is considered when enumerating all the $\Bigi$ item partitions, this allocation will be stored in $\ASet.$ Hence, the allocation returned has value at least $(1-\epsilon)$-$\MMS$ for the smallest valued bundle.

For running time, note that every iteration of the for loop first allocates all $\Sml$ chores, then either runs a bag-filling algorithm which takes $O(m)$ time, discards the iteration, or runs the $\PTAS$ of \cite{Jansen2016} which takes $O(2^{\tilde{O}(1/\epsilon)}n\log m)= o(2^{(1/\epsilon^2)}n\log m)$ time. In the worst case, every iteration takes $O(m + O(2^{1/\epsilon^2}n\log m))$ time. The for loop runs for $n^{|\bg|}$ iterations, which from Lemma \ref{lem:big-const} is $O(n^{n/\tau \epsilon})=2^{O(n\log n/\tau \epsilon)}.$ Hence, the total run time of the algorithm is $O(2^{n\log n/\tau \epsilon}(2^{1/\epsilon^2}n\log m+m))=O(m)$ time.
\end{proof}

\noindent
\textbf{Acknowledgments.} We would like to thank Prof. Jugal Garg for several valuable discussions. 

\bibliographystyle{alpha}{}
\bibliography{literature}

\appendix
\section{Missing proofs\label{sec:missing-proofs}}

\subsection{Section \ref{sec:prelims}}

\Sign*
\begin{proof}
If the sum of valuations of all items $v_i(\M)$ is negative, there can be no allocation where every bundle has non-negative valuation. Hence, $\MMS_i$ is negative. If the sum of valuations is positive, then adding all items to one bundle and no item in other bundles makes the least-valued bundle have zero value. Thus, in this case, $\MMS_i\ge 0$. 
\end{proof}

\Average*
\begin{proof}
If $\MMS_i > v_i(\M)/n$, it implies that there exists a partition of items in $\Pi_n(\M)$ where all bundles have value greater than $v_i(\M)/n$. Therefore, $v_i(\M) \ge n. \MMS_i > n.\frac{v_i(\M)}{n}  = v_i(\M),$ which is a contradiction.
%\rnote{Fill up.}
\end{proof}

\ScaleInv*
\begin{proof}
For any agent $i,$ the value of any bundle of items $\Set$ according to the two valuation function are related as $v'_i(\Set)=c_i\cdot v_i(\Set).$ Thus, by definition of $\MMS,$ her $\MMS$ values according to the two valuation functions are also related as $\MMS'_i=c_i\cdot \MMS_i,$ where $\MMS'_i$ is agent $i'$s $\MMS$ value according to $v'_i.$ 

This implies that a set of items has $\aMMS$ value for $i$ according to $(v_i)_{i\in\N}$ if and only if it has $\aMMS$ value according to $(v'_i)_{i\in\N}.$ Hence all $\aMMS$ allocations according to valuations $(v_i)_{i\in\N}$ are also $\aMMS$ according to $(v'_i)_{i\in\N}$ and vice versa.  

The $\aMMS+\PO$ allocation according to $(v_i)_{i\in\N}$, say $\A,$ is $\aMMS$ and must also be $\PO$ according to $(v'_i)_{i\in\N},$ as otherwise the Pareto dominating allocation will Pareto dominate $\A$ according to $(v_i)_{i\in\N}$ too.   
\end{proof}

\subsection{Section \ref{sec:mixed-pos}}
\clmMuTmu*
\begin{proof}
We know from condition $2$ of the $\aMMS+\PO$ problem that $v_i(\M)\ge \tau\cdot \min\{v_i^+,v_i^-\}.$ After scaling, we have $|v_i(\M)|=n.$ For agents where $v_i(\M)\ge 0,$ $n\ge \tau\cdot v_i^- \Rightarrow v_i^-\le n/\tau=O(n).$ Also, $n=v_i^+-v_i^-\Rightarrow v_i^+\le n(1+1/\tau)=O(n).$ Analogously we prove the claim when $v_i(\M)<0.$
\end{proof}

\lemNidBgGoodChore*
\begin{proof}
Since $\bg = \cup_{i\in \N} \bg^+_i \cup_{i\in \N} \bg^-_i$, to prove the lemma it suffices to show that the number of $\bg$ goods and chores of every agent $i\in \N$ is $|\bg_i^+|, |\bg_i^-|\le O(n^2/\epsilon)$. Fix an agent $i\in \N$. First we will show bound on $|\bg_i^+|$.

\noindent{\em Case 1: $\tmu_i\ge 0$.} %First we show the bound when $\tilde{\mu}_i\ge 0$.
% \vspace{-1em}
\begin{equation}
\label{eq:big+in-non-ident}
\text{If $\tmu_i\ge 1/3,$ }|\bg_i ^+| \le |\{j \in \M : v_{ij} > \epsilon/(6n)  \} | \le \frac{v_i^+}{\epsilon/(6n)} \le \frac{6n}{\epsilon}O(n) = O(n^2/\epsilon)
\end{equation}
The last inequality follows by Claim \ref{clm:vi_are_n}.
Otherwise, if $\tmu_i < 1/3 $, then $\mu_i\le \tmu_i/(1-\epsilon/2)<1/(3-3\epsilon/2)< 2/3.$ Divide $\bg^+_i$ into two sets as follows. 
\begin{equation}
\label{eq:bigs-small-mu}
    \bg^+_i = \{j: v_{ij } > \epsilon/(6n) \} \cup \{j: \epsilon \tmu_i/(2n) <  v_{ij} \le \epsilon/(6n) \} .
\end{equation}
Let us call the first set in Equation \eqref{eq:bigs-small-mu} $\bl$ and the second set $\bs$. Similarly as for the case of $\tmu_i\ge 1/3$, we can prove the size of $\bl$ is at most $O(n^2/\epsilon)$. We now prove that the number of items in $\bs$ is at most $\frac{2n(n-1)}{(1-\epsilon/2)\epsilon}+(n-2)$. We show that if this is not true then there is a partition of all the items where all parts have value strictly more than $\mu_i$ for agent $i$ which is a contradicts that $\mu_i$ being her $\MMS$ value. %; in other words, the $\MMS$ value of agent $i$ is strictly larger than $\MMS_i$. 
The partition is as follows. Add all items except the goods from $\bs$ to the first bundle. If $i$'s value for this bundle is more than $1$, divide $\bs$ to make $n-1$ bundles with at least $2n/((1-\epsilon/2)\epsilon)+1$ items in each. Then all the remaining bundles have value at least, $$ \left( \frac{2n}{(1-\epsilon/2)\epsilon} +1\right) \left( \frac{\epsilon}{2n} \right) \tmu_i > \frac{\tmu_i}{(1-\epsilon/2)} \ge \mu_i. $$  %\Seti{define $\mu_i$}

If the value of the first bundle is less than $1$ for $i$,  then we add enough goods from $\bs$ to each bundle one by one (first bundles and all remaining empty bundles) so that their value is at least $2/3 > \mu_i.$ Since every item in $\bs$ has value at most $\epsilon/(6n) $, the value of the each bundle is less than $ 2/3$ before adding the last item and less than $2/3 + \epsilon/(6n) < 1 $ later. As each bundle's value is at most $1$ and $v(\M)=n,$ there are enough items to make $(n-1)$ bundles, each of value at least $2/3$ which is greater than $\mu_i$. This is a contradiction to definition of $\MMS_i$.

Therefore, $|\bs|\le \frac{2n(n-1)}{(1-\epsilon/2)\epsilon}+(n-2)=O(n^2/\epsilon)$. Hence $|\bg^+_i|=|\bl|+|\bs|=O(n^2/\epsilon)$, for any $i$ with $\tmu_i\ge 0$.
\medskip

\noindent{\em Case 2: $\tmu_i<0$.} Then by the definition of a $\bg$ good for this case, $|\bg^+_i|\le v_i^+/(\epsilon/(2n))=(2n/\epsilon)O(n)=O(n^2/\epsilon).$

Next we show the bound on $|\bg_i^-|$. 
By definition of a $\bg$ chore, $|\bg_i^-|\le 2n\cdot v_i^-/\epsilon = O(n^2/\epsilon),$ as from Claim \ref{clm:vi_are_n} we have $v_i^-\le O(n)$. 
\end{proof}

\lemAcyclic*
\begin{proof}
First we show how to eliminate one cycle, say $C,$ in the allocation graph of $x$. That is, we define a new allocation $x',$ that removes one cycle without reducing the value of any agent. Let there be $k$ agents and $k$ items in $C,$ with the edges as,
$$ a^1 ~\mbox{---}~ o^1 ~\mbox{---}~ a^2 ~\mbox{---}~ o^2 ~\mbox{---}~ \cdots~\mbox{---} ~ o^{i-1} ~ \mbox{---}~  a^i ~\mbox{---}~ o^i ~\mbox{---}~ a^{(i+1)} \cdots ~\mbox{---}~ o^{k-1} ~\mbox{---}~ a^k ~\mbox{---}~ o^k ~\mbox{---}~ a^1.$$ 
Each agent $a^i$ is partially assigned items $o^i$ and $o^{i-1}$ (by setting $0 \equiv k$) and each item $o^i$ is partially assigned to agents $a^{i}$ and $a^{i+1}$ (by setting $k+1 \equiv 1$). Without loss of generality, we may assume that items in $C$ are solely considered as either a good or a chore by both agents sharing them, otherwise, we can break the cycle by allocating the share of the other agent for this item to the one who considers it as a good. We call an item a good if both the agents sharing it consider it so, else a chore.

First we argue the case when there is at least one good. Without loss of generality we assume $o^k$ is a good. Let $X^C = [x_{11}, x_{21}, x_{22}, \dots ,x_{k(k-1)} ,x_{kk}, x_{1k}]$ be the allocation vector of cycle $C$. Also let $\tilde{V}^C=[\tilde{v}_{11}, \tilde{v}_{21}, \tilde{v}_{22}, \dots ,\tilde{v}_{k(k-1)}, \tilde{v}_{kk}, \tilde{v}_{1k}]$ be the vector representing the \emph{scaled} values of the agents for the items assigned to them in $C,$ defined as, $$\tilde{v}_{ij}=\begin{cases}
 v_{ij} & \text{if } i=1 \\
 v_{ij} \left( \frac{v_{(i-1)(i-1)}}{v_{i(i-1)}} \right) & \text{otherwise.}
 \end{cases}$$
%Note that the $2k$ elements in $X^C $ and $\tilde{V}^C$ have the same order as it is shown in above \ST{Figure?}. 
In $\tilde{V}^C$ the valuations of the agents are scaled in a way so that agents sharing an item in $C$ have the same value for that item (except for item $k$). Without loss of generality we assume $\tilde{v}_{1k} \le \tilde{v}_{kk}$. Let $U^C = [u_{11}, u_{21}, u_{22}, \dots ,u_{k(k-1)} ,u_{kk}, u_{1k}]$ be the utility vector of $C$ where $u_{ij}= \tilde{v}_{ij}x_{ij}$. Let $\delta$ be the minimum of smallest positive $u_{ij}, i \neq j$ (even indexes of $U^C$) and smallest $|u_{ii}|, u_{ii}< 0$ (odd indexes of $U^C$). Define $${u'}_{ij}:= \begin{cases}
u_{ij} - \delta & \text{if } i\neq j \\
u_{ij} + \delta & \text{otherwise.} \\
\end{cases}$$
Then the desired $x'$ is defined as, $$x'_{ij}:= \begin{cases}
{u'}_{ij}/\tilde{v}_{ij}  & \text{if } i,j \in C \\
x_{ij} & \text{otherwise. }
\end{cases}$$

By choice of $\delta$, at least one $x'_{ij}$ with $x_{ij} >0$ will be 0 and no new edge is added to the allocation graph so the cycle $C$ is removed. We need to show that the new $x'_{ij}$'s present a feasible allocation. By choice of $\delta$, we can see that $u'_{ij} \ge 0$ when $j$ is a good for agents sharing it in $C$, $u'_{ij} \le 0$ otherwise.  For all agents $a^i$ in $C$, $u_{i(i-1)}+u_{ii} = u'_{i(i-1)}+u'_{ii}$ (by setting $1-1 = k$ for agent $a^1$) so each agent will get the same utility before removing the cycle. Also, we have $u_{ii}+u_{i(i+1)} = u'_{ii}+u'_{i(i+1)}$ and since for all items $O^i, i\in [k-1]$, $\tilde{v}_{ii} = \tilde{v}_{i(i+1)}$ we have $x_{ii}+x_{i(i+1)} = x'_{ii}+x'_{i(i+1)}$. For item $k$ we have, \begin{align*}
    &u'_{kk} = u_{kk} + \delta  \implies  x'_{kk} =  x_{kk} +  \frac{\delta}{\tilde{v}_{kk}} \\ 
    &u'_{1k} = u_{1k} - \delta \implies  x'_{1k} =  x_{1k} -  \frac{\delta}{\tilde{v}_{1k}} \\
    \implies &  x'_{kk} +  x'_{1k} \le x_{kk} +  x_{1k} \ .
\end{align*}
The last inequality hold because $\tilde{v}_{1k} \le \tilde{v}_{kk}$. Therefore, all agents receive the same utility in the new allocation. But there may be an extra amount of good $k$ available; we assign it to the agent who has the highest share of good $k$.

If all items in $C$ are chores, we define $\tilde{V}$ and $U$ similarly as for the previous case. Without loss of generality, we assume $\tilde{v}_{1k} \le \tilde{v}_{kk}$ and we choose $\delta$ to be the smallest $|u_{ii}|, u_{ii}< 0$ (odd indexes of $U^C$). With the same analysis we get $u_{ii}+u_{i(i+1)} = u'_{ii}+u'_{i(i+1)}$ for all agents $i$, $x_{ii}+x_{i(i+1)} = x'_{ii}+x'_{i(i+1)}$ for items $i \neq k$ and $ x'_{kk} +  x'_{1k} \ge x_{kk} +  x_{1k}$. Therefore, agents get the same utility with an extra amount of chore $k$ assigned to some agent. We improve the utility of the agent who gets this share of chore $k$ by reducing her share from chore $k$ by making, $\sum_{i \in \N} x_{ik} =1$.

We repeat this process for every cycle, removing at least one edge with every removal. Hence, in polynomial time, we get an acyclic allocation graph.
\end{proof}

\lemShared*
\begin{proof}
Suppose there are $k$ shared goods. Consider the subgraph of the allocation graph with the $n+k$ nodes corresponding to all the buyers and only the shared goods. As this graph is acyclic, there are at most $n+k-1$ edges. Further, each item is shared, meaning there are at least two edges incident to each node representing a good. Thus, there are at least $2k$ edges. The inequality $n+k-1\ge 2k$ is satisfied only when $k\le n-1$, hence there are at most $n-1$ shared goods.  
\end{proof}

\CorrAMMS*
\begin{proof}
If an $\alpha$-$\MMS$ allocation exists, then for the partition of $\bg$ corresponding to this allocation, say $\PB=[B_1\cdots,B_n],$ there is an integral allocation of $\nbg$ where every agent $i$ gets value $\alpha\cdot\tmu_i-v_i(B_i)\ge c_i$ from $\nbg.$ Thus, the LP will have a (fractional) solution. From Lemma \ref{lem:aemms}, the resulting allocation obtained by rounding the LP solution is $\aeMMS.$
\end{proof}

\thmNonId*
\begin{proof}
From Corollaries \ref{corr:aemms-when-exists} and \ref{corr:gpo-when-exists} the correctness of Algorithm \ref{algo:nonid-mixed-plus} follows. Next we analyze the running time. 

The time to compute the approximate $\MMS$ values is $O(n\cdot 2^{(n\log n)/\epsilon}(2^{1/\epsilon^2}n\log m+m)),$ from the proofs of Theorems \ref{thm:main-id-pos} and \ref{thm:main-id-neg}. Since $|\bg|\le O(n^3/\tau \epsilon)$ by Lemma \ref{cor:big-const}, the number of iterations in the for loop enumerating all the allocations of the $\bg$ items is $O(2^{O((n^3\log n)/\epsilon)})$. Note that we re-define $\eps$ as $\min\{\eps,\frac{\alpha\gamma}{(1+\gamma)}\}\ge \min\{\eps,\frac{\gamma^2}{2}\} =:\zeta,$ thus $|\bg|\le O(n^{3}/\zeta).$ Each iteration solves an LP of $mn$ variables and $O(mn)$ constraints, hence takes time some polynomial function in $(m,n)$ less than $O((mn)^{3})$~\cite{LeeSZ19}. Finding a cycle in the allocation graph requires time linear in the number of edges, at most $O(mn).$ Eliminating the cycle requires time $O(mn),$ and deletes at least one edge. Repeating the process until the graph is acyclic takes at most $O(mn)$ iterations, hence the making the allocation acyclic and rounding it steps take time at most $O(m^2n^2)$. Hence the total time for the algorithm in the worst case is,
$$O(n\cdot 2^{n\log n/\epsilon}(2^{1/\epsilon^2}n\log m+m))+ O(2^{O(n^3\log n/\zeta)}m^3n^3+m^2n^2)\le O(2^{O((n^3\log n)/ \min\{\eps^2,\zeta\})}m^3),$$
$$=O(2^{O((n^3\log n)/ \min\{\eps^2,\gamma^2/2\})}m^3)=O(m^3),$$ as $n,$ $\alpha,\ \gamma$ and $\epsilon$ are constant. 
\end{proof}

\section{Non-existence of $\aMMS$ allocations}
\label{sec:nonexist}

In this section, we show an instance for which there is no $\aMMS$ allocation for any $\alpha>0$. Our instance is a modification of the instance in \cite{kurokawa2016can} that shows that an $\MMS$ allocation in a goods only manna does not always exist. We take their exact instance, and add three chores to $\M,$ each of absolute value equal to a small constant less than the agent's $\MMS$ values. For completeness, we discuss all details of the instance.

Let $\N= \{1,2,3\}$, $\M^+ = \{(j,k): j \in [3], k \in [4] \} $, $ \M^- =\{(1),(2),(3) \}$, and $\M = \M^+ \cup \M^-$ respectively be the set of agents, goods, chores, and all items. In order to define the valuations of the agents for each of these items, we first define matrices $ O, E^{(1)}, E^{(2)}$, and $ E^{(3)}$ as follows.
\begin{equation*}
O= \begin{bmatrix}
17 & 25 & 12 &1\\
2 & 22 & 3& 28\\
11 & 0 & 21 & 23
\end{bmatrix} 
\end{equation*}
\begin{equation*}
E^{(1)}= \begin{bmatrix}
3 & -1 & -1 & -1\\
0 & 0 & 0& 0\\
0 & 0 & 0& 0
\end{bmatrix} \hspace{1cm}
E^{(2)}= \begin{bmatrix}
3 & -1 & 0 & 0\\
-1 & 0 & 0& 0\\
-1 & 0 & 0& 0
\end{bmatrix} \hspace{1cm}
E^{(3)}= \begin{bmatrix}
3 & 0 & -1 & 0\\
0 & 0 & -1 & 0\\
0 & 0 & 0 & -1
\end{bmatrix} \ .
\end{equation*}

The valuation of each agent $i$ for each good $(j,k)$ is,  $v_i(\{(j,k)\}) = 10^6 + 10^3 \cdot O_{jk} + E^{(i)}_{jk},$ and their value for each chore is $-4054999.75$. 

From \cite{kurokawa2016can}, every agent can divide all the goods in this instance into three bundles of value $4055000$ each. Adding one chore to each of these makes every bundle's value $0.25.$ It can be verified that the average value of all items is $0.25$ for every agent. As $\MMS$ cannot be higher than the average, the above allocation shows that every agent's $\MMS$ value is $0.25.$ 
\cite{kurokawa2016can} also show that there is no allocation of the goods were all agents get at least $4055000,$ and that the sum of any $3$ goods is less than $4055000$. As the values of goods are integers, every agent must get at least $4$ goods for every chore in order to receive a positive valued bundle. If every agent is to get a positive valued bundle, the agent receiving less than $4055000$ from the goods must not receive any chore, and must get at least one good. But then there are $3$ chores and at most $11$ goods remaining to be allotted. Hence, at least one agent will receive a negative valued bundle. Therefore, there is no allocation that can guarantee every agent a positive valued bundle, and the best $\alpha$ for which an $\aMMS$ allocation exists is at most zero.% ($\alpha$ is in fact equal to zero, as one $0$-$\MMS$ allocation is $A_1 = \M, A_2=A_3=\emptyset$).

%The agent who receives at least two chores must receive goods worth at least $4055000\cdot 2.$ 

%Since all valuations for goods are integer, there is no possible allocation were each agent gets one chore and all agents get a positively valued bundle. Moreover, since $v_i(\M)= 0.75$ and all goods value more than $10^6$, it is not possible to get an allocation where some bundle has more than one chore and all agents get a a positively valued bundle. Therefore, in this instance, 

\section{Computing $\MMS$ when $\MMS<0$}\label{appendix:ptas-neg-id}

In this section we introduce the algorithm that finds a $(1-\epsilon)$-$\MMS$ allocation of an agent with $\MMS<0$ for an instance $\MMSidins$ and a constant $ \epsilon>0,$ or equivalently, a $(1-\epsilon)$-$\MMS$ allocation of $\MMSins$ when there are identical agents with valuation function $v$ (Algorithm \ref{algo:MMS-value-ve}). From Lemmas \ref{lem:avg} and the normalization $v(\M)=-n,$ we have $\MMS\le -1.$ 

From Definition \ref{def:mms}, a $(1-\epsilon)$-$\MMS$ allocation gives each agent a bundle with value at least $(1/(1-\epsilon)) \MMS$. Let $\sigma := \frac{1}{1-\epsilon}-1$. Algorithm \ref{algo:MMS-value-ve} obtains an allocation where each agent gets a bundle of value at least $(1+\sigma)\MMS= (1/(1-\epsilon))\MMS.$ The high level idea of the algorithm is as follows. First we scale the valuations so that $v(\M)=-n$, and classify items as $\Bigi$ or $\Sml$. Then similarly as in Algorithm \ref{algo:MMS-value}, we enumerate all partitions of $\Bigi$. While there are unallocated $\Sml$ goods, we add them one by one to the bundle with the least value. Once all the $\Sml$ goods are exhausted, we iteratively add $\Sml$ chores to the bundle with the highest value.

\begin{algorithm}[tbh!]
\caption{$(1-\epsilon)$-$\MMS$ Allocation for identical agents with $\MMS<0$} \label{algo:MMS-value-ve}
\DontPrintSemicolon
  \SetKwFunction{Define}{Define}
  \SetKwInOut{Input}{Input}\SetKwInOut{Output}{Output}
  \Input{$\MMSidins$, a constant $\epsilon$}
  \Output{$(1-\epsilon)$-$\MMS$ Allocation}
  \BlankLine
  Normalize the valuations so that $v(\M)=-n.$ \label{line:norm-ve}\;
  $\sigma \gets \frac{1}{1-\epsilon}-1,\ \ASet\gets \emptyset $ \label{line:initial-ve} \;
  $\Bigi := \{j \in X: |v_j| \ge \sigma \}$, $\Sml\coloneqq \M\backslash \Bigi $, $\Sml^+=\Sml\cap \M^+,\Sml^-=\Sml\cap \M^-$\;
  \For{$\B \in \Pi_n (\Bigi)$ \label{line:enumerate-ve}} {
    \While {$\Sml^+\ne \emptyset$ \label{line:coverB1-ve}} {
    add any $\Sml$ good to a bundle with the lowest value
    }
    \While{$X \neq \emptyset$}{
    add any $\Sml$ chore to a bundle with the highest value
  }
     store the allocation to a set $\ASet$ \label{line:store-ve}
    }
  \Return $\argmax_{A\in \ASet} \min_{A_i \in A} v(A_i)$\tcp*{return allocation with highest maximin value} \label{line:last-ve}
\end{algorithm}

\begin{theorem}\label{thm:main-id-neg}
Algorithm \ref{algo:MMS-value-ve} gives a $(1-\epsilon)$-$\MMS$ allocation when $\MMS<0$ in $O(m)$ time.
\end{theorem}
\begin{proof}
We first prove a helpful lower bound on the value of all $\Sml$ goods. Let $\PB^*$ be a partition of the $\Bigi$ items corresponding to an $\MMS$ allocation. There are enough $\Sml$ goods to add to each part in $\PB^*$ so that every part has at least $\MMS$ value. Specifically, for the set $\Set=\{B \in \PB^*: v(B)< \MMS\}$ we have, 
\begin{equation}
    \label{eq:bound-small+}
    v(\Sml^+) \ge \MMS\cdot |\Set| - v(\PB^*)
\end{equation}
Now, let $\PA= \{A_1,\dots,A_n \}$ be the output of Algorithm \ref{algo:MMS-value-ve}. Suppose for contradiction there exists some $A_i \in A$ such that $v(A_i) < (\frac{1}{1-\epsilon})\MMS = (1+\sigma)\MMS$. Consider each $A_k \supseteq B_k$ with $B_k\in\Set$. Note that the algorithm adds $\Sml$ goods to the bundle with the least value. Because of $A_i$, before adding the last $\Sml$ good to any bundle, its value is less than $(1+\sigma)\MMS$. The last good added has value at most $\sigma$. Therefore, all the $A_k$s have value at most $\MMS$. From, \eqref{eq:bound-small+} and the fact that the algorithm adds goods to the least valued bundle, we have,
\begin{equation}
    \label{eq:bound-small-sigma}
    v (\Sml^+ \setminus ( \bigcup_{k \in [n]} A_k  ) ) \ge \sum_{B \in \Set} \MMS - v(B)+\sigma = \sigma, 
\end{equation}
which is a contradiction.

Now we prove that after adding the $\Sml$ chores the value of all the bundles is at least $(1+\sigma)\MMS$. This is true because while there exists an unallocated chore, the value of the highest valued bundle is greater than $-1,$ because $v(\M)=-n$. Adding a chore to such bundle will decrease the value by at most $\sigma$. Therefore, the value of such bundle is at least $-(1+\sigma ) \ge (1+\sigma )\MMS$. By definition of $\sigma, (1+\sigma)=1/(1-\epsilon).$

Finally, $|\bg|=O(n/\sigma)=O(n/\eps),$ from the definition of $\bg$ and $\sigma.$ As every iteration corresponding to a partition of $\bg$ takes $O(m)$ time, Algorithm \ref{algo:MMS-value-ve} runs for $O(m\cdot 2^{O(n\log n/\eps)})=O(m)$ time.
\end{proof}

\section{Hardness of Approximation}\label{sec:hardness}

The $\aMMS+\PO$ problem makes two assumptions. First, the number of agents is assumed to be a constant. Second, the sum of absolute values of all the items for every agent is assumed to be at least $\tau$ times the minimum of this sum for the goods and the chores, for some constant $\tau>0$. In this section we show that relaxing either of these two assumptions makes the $\aMMS$ problem $\classNP$-hard for any $\alpha\in (0,1],$ even when agents are identical.

When agents are identical, the allocation that decides the $\MMS$ value of the agents is also an $\MMS$ allocation for the instance. Thus, for $\alpha=1$, the $\aMMSg$ problem should return an $\MMS$ allocation. Furthermore, given $v(\M)>0,$ we are guaranteed to have $\MMS\ge 0$ due to Lemma \ref{lem:mms-sign}. However next we show that when either assumption of problem $\aMMS$ is dropped, {\em deciding} if the inequality is indeed strict is $\classNP$-hard. 

We separate Theorem \ref{thm:hard} as two $\classNP$-hardness results in Theorems \ref{thm:hard-const} and \ref{thm:hard-gc-rel}. To prove both, we reduce from the known $\classNP$-hard {\Parti} problem.

\noindent{\bf {\Parti} Problem.} Given a set of non-negative integers $E=\{e_1,\dots,e_m\}$, output YES if there exists a division of the elements into two sets of equal weight, otherwise output NO.

%Alternatively, we ask if the $\MMS$ value is at least $k$.

%For identical agents, the $\MMS$ value is the same for all agents, hence is either non-negative for all, or negative for all. 

%We prove the results when the $\MMS$ value is non-negative, and remark what modifications make the proofs work for the other case. 

%starting with a proof establishing that finding $\aMMS$ is $\classNP$-Hard for any constant $0< \alpha \le 1$ even when $|\N|=2$ and agents are identical.
%Recall Theorem \ref{thm:hard-const}:
%\genconsthard*

\begin{restatable}{thm}{genconsthard}\label{thm:hard-const}
Given an instance $\MMSidins$ with constantly many (two) identical agents and $v(\M)>0$, checking if $\MMS>0$ is $\classNP$-hard.
\end{restatable}

%\begin{theorem}\label{thm:hard-const}
%Given an instance $\MMSidins$ with constantly many (two) identical agents and $v(\M)>0$, checking if $\MMS>0$ is $\classNP$-hard.
%The $\gaMMSD$ problem for a constant number of agents is $\classNP$-Hard, for any $\beta>0$.
%\end{theorem}
%\begin{proof}
%\begin{proof}
%\addtocounter{theorem}{-1}
%This theorem can be alternatively stated as follows.
%\begin{theorem}
%It is enough to show that the problem of finding $\MMS_i$ approximately up to any constant $\alpha$ is $\classNP$-Hard.
%\end{theorem}
\begin{proof}
We reduce an instance of {\Parti} to an $\MMS$ instance $\MMSidins$ with two identical agents. 
%
%Let $E=\{e_1,e_2, ..., e_m\}$ be the set of elements given as input for {\Parti}. The problem is scale-free, hence %we first scale the elements to be integer valued. Then if $2\beta >1$, 
%scale $e_i$ to  $e_i*\lceil2\beta\rceil$ for all $i\in[m].$ Each element is now at least $2\beta$.%, and the difference between any two elements is at least $2\beta$.
%
%We can assume that each element has a positive integer weight greater than $2\beta$, i.e., $ e_j\in\mathbb{N}\backslash[\lfloor2\beta\rfloor], \ \forall\  j\in[m]$, as {\Parti} is a scale-free problem. 
%
Let $\N=\{1,2\}$. $\M=[m+2]$, where the first $m$ items are goods and the last two are chores. The valuation function $v$ is defined as follows, where $\beta=1/4$.
\begin{equation*}
  v_{j}=\begin{cases}
  e_j, &\forall\ j\in [m]\\
  -(\sum_i e_i/2) + \beta, & j\in\{m+1,m+2\}.  
  \end{cases}
\end{equation*}
 
That is, the goods correspond to {\Parti} elements, and have the same value as the weight of the element, and the chores are $\beta$ more than the negated weight of each set in an equal distribution of the elements. Note that, $(a)$ the trivial partition where all items are in the same bundle has the smaller bundle valued zero, and $(b)$ the average of values of all items is $\beta$, and $\MMS$ cannot be higher than the average (Lemma \ref{lem:avg}). Hence, $0\le \MMS\le \beta$.

%Each of value $e_j, j\in[m]$ for both the agents. We also add two chores, each of value $-(\sum_i e_i/2) + \beta,$ for any $ 0<\beta < 1$. %Correctness of the reduction follows by proving both chores will be in different bundles in any $\MMS$ allocation, and removing the chores to get an equal partition of $E$. The details are as follows. %A partition of $E$ converts to an $\MMS$ allocation by adding one chore to each bundle.
%These proofs are straightforward, hence we skip them.

We prove the correctness of the reduction in the following two claims. 
\begin{claim}\label{clm:no}
{\Parti} has a solution $\Rightarrow \MMS\ge \beta$.
\end{claim}
\begin{proof}
 Divide the goods into two bundles as per the {\Parti} solution, and add one chore to each set. This gives us two bundles of equal value $\beta$, implying that $\MMS\ge \beta$.
\end{proof}

\begin{claim}\label{clm:yes}
$\MMS>0 \Rightarrow$ {\Parti} has a solution.
\end{claim}
\begin{proof}
We prove the contrapositive by contradiction. Suppose {\Parti} does not have a solution. but $\MMS>0$ for the instance $\MMSidins$. 
%As $\MMS\ge \beta$, for the given instance $\MMS = \beta$, and there is some allocation $\A$ that has value $\beta$ for both bundles.  %Suppose {\Parti} does not have a solution. %We know that $\MMS\ge 0.$
Let $\PA=(A_1,A_2)$ be the allocation achieving the $\MMS$ value, and let $u_1=v(A_1)$ and $u_2=v(A_2)$. Then we have $u_1,u_2>0$. 

First we prove that both the chores cannot be in the same bundle. If they are, and if all goods are not in this bundle, then the value of the bundle with chores is at most the sum of all except the smallest good. This is $(-\sum_i e_i + 2\beta) + (\sum_i e_i - \min_i e_i) \le  1/4 - 1 < 0.$ If every good and chore is in the same bundle, the value of the other bundle is $0$. But $v_1>0$, hence the chores are in separate bundles. %A zero valued bundle is not $\aMMS$ valued for any $\alpha\in (0,1]$, hence no such allocation is an $\aMMS$ allocation.

But then the value of the goods in each bundle is at least the total value minus the chore's value, i.e., for $i=1,2$, $v(A_i \cap \M^+) = u_i - (-\frac{1}{2} \sum_{i\in[m]} e_i + \beta) \ge \MMS - \beta +\frac{1}{2}\sum_i e_i > \frac{1}{2}\sum_i e_i - \beta$. Since $\beta=1/4$ while $v(A_i \cap \M^+)$ and $\frac{1}{2}\sum_i e_i$ are integers, it follows that $v(A_i\cap \M^+)\ge \frac{1}{2}\sum_i e_i$. Then partition $(A_1\cap \M^+, A_2\cap \M^+)$ of $E=(e_1,\dots,e_m)$ is a solution of the {\Parti} problem, a contradiction.
\end{proof}
 Claims \ref{clm:no} and \ref{clm:yes} show that $\MMS >0 \iff $ there is a solution to {\Parti}. %Theorem \ref{thm:hard-const}.  
%\end{proof}

%\begin{remark}
%The complementary problem for the case when the sum of values is negative asks the same question for $\MMS<0.$ This problem too is $\classNP$-hard, and we later sketch the modification that make the proof work for this setting. 
%\end{remark}

%\begin{remark}
%When $v(\M)<0,$ a similar reduction shows the $\aMMS$ problem is $\classNP$-Hard for $\alpha\in(1/2,1].$ We swap the signs of all the partition elements, and add goods with value $\sum_j e_j/2 -\beta$. Then we show that $\MMS=-\beta$ if $\Parti$ has a solution, $\MMS=-2\beta$ otherwise. Hence, a better than $2-$approximate $\MMS$ allocation would be able to distinguish between the two cases. 
%The same reduction also shows that Theorem \ref{thm:hard-const} holds for the case when $\MMS$ is negative for $\beta=-1/4$. As the average of all valuations is $\beta$, $\MMS\le \beta$. We prove $\MMS<0$ if and only if {\Parti} has a solution.
%\end{remark}

%\subsection{Additive valuations}

%vi(M^+)>=r.vi(M^-)

When agents are identical, they agree on every item if it is a good or a chore, and therefore $\M^{gc}=\emptyset$. Therefore, $v^+_i$ and $v^-_i$ as defined in Definition \ref{def:amms} are same as $v(\M^+)$ and $|v(\M^-)|$ respectively.
\end{proof}

%\genRelGcHard*

\begin{restatable}{thm}{genRelGcHard}\label{thm:hard-gc-rel}
Given a fixed constant $\tau>0$, even if an instance $\MMSidins$ with identical agents satisfies $|v(\M)| \ge \tau\cdot\min\{ v(\M^+),|v(\M^-)|\}$, checking if $\MMS>0$ is $\classNP$-hard.
%The $\gaMMSD$ problem is $\classNP$-Hard for any $\beta>0$, even when the valuation function satisfies $v(\M^+)\ge (1+\tau) v(\M^-)$, for some constant $\tau>0$. 
\end{restatable}
\begin{proof}
Again, we give a reduction from {\Parti}. Let $E=\{e_1,e_2, \cdots, e_m\}$ be the set of elements given as input for {\Parti}. %Again we scale the elements so that $e_j \ge 2\beta, \forall\ j\in[m]$. 
Create an instance $\MMSidins$ as follows: $\N$ has $n$ agents, where $n$ will be fixed later based on the value of $\tau$. $\M=\{1,2,\cdots,m+n\}$ where the first $m + (n-2)$ items are goods, and the last $2$ are chores. The valuation function $v$ is defined as follows, where $\beta=1/4$.
\begin{equation*}
    v_{j}=\begin{cases}
    e_j &\forall j\in [m]\\
    \beta &\text{for }j\in \{m+1,...,m+(n-2)\}\\
    -\sum_{i\in[m]} e_i/2 + \beta &\text{for } j\in\{m+n-1,m+n\}. 
    \end{cases}
\end{equation*}

That is, the first $m$ goods have values equal to the weights of the corresponding elements of {\Parti}. The remaining $(n-2)$ goods have value $\beta$ each, and both the chores have value $-(\sum_i e_i/2) + \beta.$ Fix $n$ to satisfy $|v(\M)| \ge \tau\cdot \min\{ v(\M^+),|v(\M^-)|\}$, or equivalently $v(\M^+)\ge (1+\tau) |v(\M^-)|,$ that is, $((n-2)\beta + \sum_i e_i)\ge (1+\tau) (\sum_i e_i +2\beta)$. 

We again have $0\le \MMS\le \beta.$ The lower bound because $v(\M)>0$ and Lemma \ref{lem:mms-sign}, and the upper bound because the average $v(\M)/n$ is $\beta$ and Lemma \ref{lem:avg}. The correctness is argued in the next two claims. 

\begin{claim}\label{clm:no-gc}
{\Parti} has a solution $\Rightarrow \MMS\ge \beta.$
\end{claim}
\begin{proof}
Divide the first $m$ goods as per the division of the elements of {\Parti} into equal valued sets, and add one chore to each bundle. From the remaining goods $\{m+1,...,m+(n-2)\}$ give one each to the remaining $(n-2)$ bundles. The value of every bundle created is $\beta$. Hence, $\MMS\ge \beta$.
\end{proof}

\begin{claim}\label{clm:yes-gc}
 $\MMS>0 \Rightarrow $ {\Parti} has a solution.
\end{claim}
\begin{proof}
We prove the contrapositive of the statement, by contradiction. Suppose {\Parti} instance $E={e_1,\dots,e_m}$ does not have a solution, but $\MMS>0$ for $\MMSidins$. 

Given that there are exactly two chores, at least $(n-2)$ bundles have only goods and has to have at least one good. Furthermore, since $e_i$s are positive integers and $\beta=1/4$, each of these $(n-2)$ bundles have value at least $\beta$. Now, $\beta$ being the upper bound on the $\MMS$ value, wlog we can assume that these $(n-2)$ bundles have exactly one good of the minimum value, namely $\beta$. This exhaust the goods $\{m+1,\dots,m+(n-2)\}$ with value $\beta$. Therefore, the two chores and all goods corresponding to the {\Parti} problem elements, and no other good, are in the remaining two bundles. Let these be the first two bundles $A_1$ and $A_2$. 

Now by the same argument as in the proof of Claim \ref{clm:yes}, we can show that both $A_1$ and $A_2$ have positive value only if each contains exactly one chore and the total value of goods in each, namely $v(A_i\cap E)$ for $i=1,2$, is at least $\frac{1}{2} \sum_{i\in[m]} e_i$. Thus, $(A_1\cap E, A_2\cap E)$ is a solution to the {\Parti} instance $E$, a contradiction.
\end{proof}

Claims \ref{clm:no-gc} and \ref{clm:yes-gc} show that $\MMS >0$ for $\MMSidins$ $ \iff $ there is a solution to {\Parti}. %Theorem \ref{thm:hard-const}.  
%From Claims \ref{clm:no-gc} and \ref{clm:yes-gc}, we get that $\MMS\ge \beta \iff $ {\Parti} has a solution, hence proving Theorem \ref{thm:hard-gc-rel}.
%
%Thus, if the $\aMMS$ allocation has a zero valued bundle, {\Parti} does not have a solution. Else, the division of  the elements of {\Parti} according to the $\aMMS$ allocation gives the solution to {\Parti}.
\end{proof}

%\begin{remark}
%If we change the relation of valuation of goods and chores to $|v(\M^-)|\ge (1+\tau)v(\M^+)$, the $\aMMS$ problem is $\classNP$-hard for $\alpha\in(1/2,1]$. A similar reduction as for Theorem \ref{thm:hard-gc-rel} shows this. The modifications are we swap the signs of the {\Parti} elements and corresponding items, add two goods of value $\sum_j e_j/2 -\beta,$ and $(n-2)$ \textit{chores} of value $-\beta$. We show $\MMS=-\beta$ if {\Parti} has a solution, and $-2\beta$ otherwise. 
%\end{remark}

Theorems \ref{thm:hard-const} and \ref{thm:hard-gc-rel} show that even if we know that $\MMS\ge 0$ checking if it is strictly positive is $\classNP$-hard. Since for $\alpha\in(0,1]$, $\MMS>0 \Leftrightarrow \alpha\MMS>0$, this essentially means, we can not find an $\alpha$-$\MMS$ allocation for {\em any} value of $\alpha \in (0,1]$ if either of the two conditions in $\aMMS$ problem is dropped. The next theorem formalizes this.

\GenMMSHard*

Even though an instance with identical agents is guaranteed to have an allocation where every agent gets at least the $\MMS$ value, i.e., $1$-$\MMS$ allocation exists, Theorem \ref{thm:hard} ruling out an efficient algorithm for finding $\alpha$-$\MMS$ allocation any $\alpha\in (0,1]$ is very striking. In light of this result, it is evident that even getting a $\PTAS$, in other words finding $(1-\epsilon)$-$\MMS$ allocation, in case of identical agents is non-trivial and important. 

\begin{comment}
\section{nonidentical agents}
is this required? missing stuff:
-preprocessing
-full algorithm formally
-opt-alpha-mms formally
\end{comment}
\section{Detailed Related Work }\label{sec:relWork}

\noindent
\textbf{Fairness and efficiency in mixed manna.} While ours is the first work on $\MMS+\PO,$ finding fair and efficient allocations has been studied for other notions. ~\cite{AzizCIW19} initiate the study for a mixed manna, and study the problem of finding $\EF+\PO$ allocations. ~\cite{AleksandrovW20} study fairness properties related to $\EFX,$ defined as envy-freeness up to any item along with $\PO.$ 

\noindent
\textbf{Fairness for Mixed Manna.}
Finding fair allocations of mixed items has recently caught a lot of attention for both divisible \cite{bogomolnaia2017competitive,bogomolnaia2019cbads} and indivisible \cite{Aleksandrov2019mixed,aleks2020algorithms, Aleksandrov20, GargM20}, items. However, to the best of our knowledge, ours is the first study on $\MMS$ allocations for a mixed manna.

\noindent
\textbf{Fairness and efficiency in goods manna.} This problem is well-studied for a goods manna. Two popular notions for a goods manna are the Nash social welfare ($\NSW$), and $\EF+\PO,$ defined and discussed below. 

\textit{$\NSW$.} Nash Social Welfare ($\NSW$) is the geometric mean of the valuation of the agents. The $\NSW$ problem is to find an allocation of indivisible items that maximizes $\NSW$. This problem is $\classAPX$-Hard~\cite{GargM19}, and remarkable approximation results for the linear valuations case have been proven by a connection of the problem with markets~\cite{ColeG15,ColeDGJMVY17,BarmanKV18,CheungCGGHM18} or real stable polynomials~\cite{AnariGSS17}. The best known result is a $1.45$ approximation factor~\cite{BarmanKV18}. Similar results are known, again by exploiting the market connection, for popular valuation functions like budget-additive~\cite{GargHM18}, separable piece-wise linear concave (SPLC)~\cite{AnariMGV18}, and their combination~\cite{CheungCGGHM18}. Recent results give an $O(n)$ approximation when agents have subadditive valuations, a far more general class than all the earlier ones~\cite{BarmanBKS20, ChaudhuryGM20}. Recent work has also been done on the general version of the problem with asymmetric agents, where the aim is to maximize the weighted geometric mean, for given weights, and submodular utilities~\cite{GargKK20}. This notion is not applicable for a mixed manna. 

\textit{$\EF+\PO$.}
\textit{$\EF$} was first introduced by \cite{budish2011combinatorial} as an relaxation of envy-freeness. An allocation is $\EF$ if for any two agents $i_1$ and $i_2$, agent $i_1$ prefers (or equally likes) her own bundle to agent $i_2$'s bundle after removing \emph{some item} from the bundle of agent $i_2$. An $\EF$ allocation can be found efficiently using envy cycle removal procedure introduced by \cite{lipton2004approximately}. ~\cite{BarmanKV18} show a pseudo-polynomial time algorithm to obtain an $\EF+\PO$ allocation on a goods manna. A series of works ~\cite{AzizMS20,Zeng2020,chaudhury2020dividing,SandomirskiyS19} study special cases of the problem.  

\textit{Other notions} studied for a goods manna are $\Prop+\PO$~\cite{AzizMS20}), or group fairness notions~\cite{Conitzer19}. When the preferences are ordinal, ~\cite{AzizHS19} discuss $\EF$ solutions that satisfy the efficiency notions of utilitarian maximality and rank maximality.

\noindent
\textbf{$\MMS$.} The study of fair division started with the cake cutting problem \cite{steinhaus1948problem}. Two popular notions of fairness established here were proportionality, meaning each agent must get a bundle worth at least $1/|\N|$ of her value for all items, and envy-freeness, where each agent must value her own bundle at least as much as any other. However, neither of these can always be attained when the items are indivisible. A simple example is allocating one good between two agents; there is no allocation that is proportional or envy-free. This motivated the search for new fairness notions for indivisible items. One well-studied notion resulting from this investigation is $\MMS$ \cite{budish2011combinatorial}. In recent years, the problem of finding $\MMS$ allocations gained a lot of interest, and a series of impressive results were found for various special cases of the problem, as discussed below. 

\noindent
\textbf{$\MMS$ for Goods.} \cite{BouveretL16} showed that in some restricted cases $\MMS$ allocations always exist. A notable result from \cite{ProcacciaW14} showed that $\MMS$ allocations may not always exist but $2/3$-$\MMS$ allocations always do. A series of works studied the efficient computation of $2/3$-$\MMS$ allocations for any $n$ \cite{amanatidis2017approximation,BarmanK17,garg2018approximating}. ~\cite{ghodsi2018fair} showed that a $3/4$-$\MMS$ allocation always exists. Most recently \cite{garg2019improved} showed that a ($3/4+1/(12n)$)-$\MMS$ allocation always exists. Finding $\MMS$ values is hard but a $\PTAS$ for this problem is known \cite{woeginger1997polynomial}. This $\PTAS$ can be used to find a ($3/4+1/(12n)-\epsilon$)-$\MMS$ allocation for $\epsilon>0$ in polynomial time. There is also a strongly polynomial time algorithm to find $3/4$-$\MMS$ allocation \cite{garg2019improved}. Other notable works on the goods only case before being improved by follow-up work are ~\cite{FarhadiGHLPSSY19,garg2018approximating,kurokawa2016can,KurokawaPW18}.

%\cite{amanatidis2017approximation,BarmanK17,BouveretL16,FarhadiGHLPSSY19,garg2018approximating,ghodsi2018fair,kurokawa2016can,ProcacciaW14,KurokawaPW18,garg2019improved}. 

\noindent
\textit{Constant number of agents with a goods only manna.} For three agents, \cite{amanatidis2017approximation}  showed that a $7/8$-$\MMS$ allocation always exists. This factor was later improved to $8/9$ in~\cite{GourvesM19}. For four agents, \cite{ghodsi2018fair} showed that a $4/5$-$\MMS$ allocation always exist. 

%In all these works it has been assumed that items have non-negative valuations, i.e, $\M$ consists of all goods. There is an impressive series of works for the case when $\M$ is consist of all chores too.

\noindent
\textbf{$\MMS$ for Chores.}
\cite{AzizRSW17} first studied the $\MMS$ problem with a chores manna. They introduced an algorithm for finding $2$-$\MMS$ allocations\footnote{Our definition of $\aMMS$ for the mixed manna is consistent for agents with positive as well as negative $\MMS$ values. We define $\alpha$ as smaller than $1,$ and consider $1/\alpha$-$\MMS$ valued bundles as $\aMMS.$ Prior results for the chores manna have $\alpha>1$ and ask for $\alpha\cdot\MMS$ valued bundles. We state the approximation factors as defined in the original papers, and ask the reader to invert them when relating with ours.}. \cite{BarmanK17} improved the previous result by showing an algorithm for a $4/3$-$\MMS$ allocation. Later, \cite{HuangL19} improved this result to a $11/9$-$\MMS$ allocation. They also showed a $\PTAS$ to find ($11/9+\epsilon$)-$\MMS$ allocation and a polynomial time algorithm to find a $5/4$-$\MMS$ allocation.

\noindent
\textbf{Other variants of $\MMS$.} The $\MMS$ problem has been studied under various other models in the goods only setting like with asymmetric agents \cite{FarhadiGHLPSSY19}, group fairness~\cite{barman2018groupwise,chaudhury2020little}, beyond additive valuations~\cite{BarmanK17, ghodsi2018fair,li2018fair}, in matroids \cite{GourvesM19}, with additional constraints ~\cite{GourvesM19,biswas2018fair}, for agents with externalities~\cite{BranzeiMRLJ13,AhmadiPourAnariEGHIMM13}, with graph constraints~\cite{bei2019connected,lonc2019maximin}, and with strategic agents \cite{strategicagents}. In the chores only setting too, weighted $\MMS$ \cite{aziz2019weighted}, and asymmetric agents \cite{aziz2019maxmin} notions have been investigated.

%\noindent 
%\textbf{Fairness and efficiency for Mixed Manna.}\RK{cite all works.}

 %\ST{Approximation Algorithms for Max-Min Share Allocationsof Indivisible Chores and Goods, aziz et al initiated it %The other notions are Formalized envy free up to one item (EF1) allocation and proportional up to one item (PROP1) allocation \cite{aziz2019fair}. They show algorithms to find EF1 allocation and PROP1 allocation.

%\noindent
%Finally, another fairness notion popularly studied in the goods manna setting is\textit{$\EFX$.} Envy free up to any good ($\EFX$) is a stronger relaxation of envy-free allocation which was first introduced in \cite{Caragiannis2016EFX}. An allocation is $\EFX$ if for any two agents $i_1$ and $i_2$, agent $i_1$ prefers (or equally likes) her own bundle to agent $i_2$'s bundle after removing \emph{any item} from the bundle of agent $i_2$. \cite{Plaut2018efx} showed the existence of $\EFX$ allocation for two restricted cases (a) when $n=2$ and (b) when agents are identical. Until recently, the existence of an $\EFX$ allocation was unknown even for $3$ agents. \cite{Chaudhury2020EFX} showed that for three agents $\EFX$ allocation always exist. The partial $\EFX$ allocation and approximate $\EFX$ allocation has also been vastly studied \cite{Caragiannis2019EFXNSW,chaudhury2020little}.

\end{document}